\documentclass[pdflatex,sn-mathphys-num]{sn-jnl}
\usepackage{graphicx}%
\usepackage{multirow}%
\usepackage{amsmath,amssymb,amsfonts}%
\usepackage{amsthm}%
\usepackage{mathrsfs}%
\usepackage[title]{appendix}%
\usepackage{xcolor}%
\usepackage{textcomp}%
\usepackage{manyfoot}%
\usepackage{booktabs}%
\usepackage{algorithm}%
\usepackage{algorithmicx}%
\usepackage{algpseudocode}%
\usepackage{listings}%

\theoremstyle{thmstyleone}%
\newtheorem{theorem}{Theorem}
\newtheorem{proposition}[theorem]{Proposition}%

\theoremstyle{thmstyletwo}%
\newtheorem{example}{Example}%
\newtheorem{remark}{Remark}%
\newtheorem{lemm}{Lemma}

\theoremstyle{thmstylethree}%
\newtheorem{definition}{Definition}%

\newcommand{\R}{\mathbb R}
\newcommand{\BE}{\begin{equation}}
\newcommand{\EE}{\end{equation}}
\newcommand{\sF}{{\mathscr{F}}}
\newcommand{\argmax}{\mathop{\rm argmax}}
\newcommand{\argmin}{\mathop{\rm argmin}}

\newcommand{\cH}{{\mathcal{H}}}
\newcommand{\cA}{{\mathcal{A}}}
\newcommand{\bsb}[1]{\boldsymbol{#1}} 
\newcommand{\ZZ}{{\mathbb Z}}
\DeclareMathOperator{\e}{e}
\newcommand{\N}{{\mathbb{N}}}

\begin{document}

\title[Article Title]{Questioning Normality: A study of wavelet leaders distribution}


\author[1]{\fnm{Wejdene} \sur{Ben Nasr}}\email{wejdene.nasr-ben-hadj-amor@u-pec.fr}
\equalcont{These authors contributed equally to this work.}

\author*[2,3]{\fnm{Hélène} \sur{Halconruy}}\email{helene.halconruy@telecom-sudparis.eu}
\equalcont{These authors contributed equally to this work.}

\author[1]{\fnm{Stéphane} \sur{Jaffard}}\email{jaffard@u-pec.fr}
\equalcont{These authors contributed equally to this work.}

\affil*[1]{\orgdiv{LAMA}, \orgname{Univ Paris Est Creteil, Univ Gustave Eiffel CNRS}, \orgaddress{ \city{Cr\'eteil}, \postcode{F-94010}, \country{France}}}

\affil[2]{\orgdiv{SAMOVAR}, \orgname{T\'el\'ecom SudParis}, \orgaddress{ \city{Évry-Courcouronnes}, \country{France}}}

\affil[3]{\orgdiv{Modal'X}, \orgname{Nanterre Université}, \orgaddress{ \city{Nanterre}, \country{France}}}


\abstract{The motivation of this article is to estimate multifractality classification and  model selection parameters :  the first-order scaling exponent $c_1$ and the second-order scaling exponent (or intermittency coefficient) $c_2$. These exponents are built on  wavelet leaders, which therefore constitute fundamental tools in applied multifractal analysis. While most estimation methods, particularly Bayesian approaches, rely on the assumption of log-normality, we challenge this hypothesis by statistically testing the normality of log-leaders. Upon rejecting this common assumption, we propose instead a novel model based on log-concave distributions. We validate this new model on well-known stochastic processes, including fractional Brownian motion, the multifractal random walk, and the canonical Mandelbrot cascade, as well as on real-world marathon runner data.  
Furthermore, we revisit the estimation procedure for $c_1$, providing confidence intervals, and for $c_2$, applying it to fractional Brownian motions with various Hurst indices as well as to the multifractal random walk. Finally, we establish several theoretical results on the distribution of log-leaders in random wavelet series, which are consistent with our numerical findings.  }




\maketitle

\section{Introduction}\label{sec1}

The motivation behind this paper is to estimate regularity parameters, namely $c_1$ and $c_2$, which represent the first- and second-order scaling exponents. These parameters are crucial for tasks such as signal and texture classification, with applications in biomedicine, as well as for analyzing fully developed turbulence. They also play a key role in selecting suitable mathematical or stochastic models that exhibit multifractal behavior in fields such as finance, geophysics, and fluid dynamics.

\noindent
A common approach involves directly estimating local regularity, often using fractional processes, including extensions of fractional Brownian motion (fBm) to capture local variations \cite{cohen2013fractional, ayache2018multifractional}. Techniques based on local increments \cite{istas1997quadratic, bertrand2013local} and wavelets \cite{ayache2012linear, ayache2015linear, jin2018estimation} have been widely explored. However, direct estimation faces key limitations, including assumptions about process types (e.g., multifractional processes), reliance on small data intervals, and instability when dealing with multiplicative cascades or L\'evy processes with erratic exponents. These challenges underscore the importance of global estimation methods, shifting the focus to multifractal analysis.

\noindent
The purpose of multifractal analysis is to estimate the sizes of sets where pointwise regularity takes specific values (referred to as the \textit{multifractal spectrum}, first defined in the seminal paper by Frisch and Parisi \cite{ParFri85}, see Section \ref{sec:ms}). This approach exploits scaling properties from log-log plots of empirical moments of order $p \in \mathbb{R}$, as originally proposed by Kolmogorov \cite{k41}. Initially restricted to $p > 0$ in \cite{ParFri85}, these methods were later extended to allow both positive and negative values of $p$, leading to methods such as WTMM (Wavelet Transform Maxima Method) \cite{muzy1991wavelets} and DFA (Detrended Fluctuations Analysis) \cite{kantelhardt2001detecting}. However, these methods lacked theoretical results linking them to the multifractal spectrum.  \\

\noindent
State-of-the-art multifractal analysis now relies on orthonormal wavelet decompositions, where multiscale quantities derived from the suprema of wavelet coefficients-known as \textit{wavelet leaders}-are used to compute scaling functions $\zeta_f$ (see Definition \ref{def:scalfun} below). Their Legendre transforms provide an upper bound for the multifractal spectrum. These methods have proven effective in applications such as model selection and classification. For instance, the estimation of moments for negative values of $p$ has been used to reject turbulence models whose scaling functions did not align with those measured from experimental data \cite{muzy1991wavelets, LASHERMES:2008:A}. However, reliable applications require a statistical framework that provides confidence intervals.  \\
An important step was the introduction by Abry and Wendt of bootstrap methods to estimate scaling exponents, improving accuracy when only a few scales are available \cite{Wendt2007b, wendt2007bootstrap,wendt2009wavelet}. In practice, rather than relying on the entire scaling function, a few coefficients from its Taylor expansion at $p=0$, computed as log-cumulants, are often sufficient. The first two cumulants, $c_1$ and $c_2$, capture key multifractal properties: $c_1$ represents the most frequently observed pointwise regularity exponent, while $c_2$ quantifies fluctuations in local regularity. Bayesian methods introduced by Combrexelle, Wendt, Dobigeon, Tourneret, McLaughlin, and Abry, under the assumption of log-normality, have significantly improved the estimation of these quantities \cite{combrexelle2015bayesian, Combrexelle2016EUSIPCO,LeonTSP2022}. The third major multifractality parameter used in classification and model selection, the \textit{uniform regularity exponent} $H^{\min}_f$, is of a different nature, as it is derived directly from wavelet coefficients rather than wavelet leaders. Its statistical estimation will be addressed in \cite{BNHHSJ2}.  \\

\noindent
Motivated by the estimation of the scaling parameters $c_1$ and $c_2$, we challenge the commonly assumed log-normality hypothesis, which is often used to justify Bayesian methods and bootstrap approaches. We explore key modeling and statistical questions: Is the log-normality assumption valid for widely studied stochastic processes? To what extent does it hold for real-world signals? When it fails, can more flexible statistical models better capture the observed distributions of wavelet leaders?  \\
To address these questions and empirically validate our statistical model, we first test for normality. We show that, in general, the log-leaders associated with standard stochastic models - such as fractional Brownian motion (fBm), the canonical Mandelbrot cascade, and the multifractal random walk - as well as those derived from real-world marathon runner data, {are not Gaussian}. This motivates us to consider a broader non-parametric framework: \textit{log-concave distributions}. We then apply a log-concavity test developed by Cule, Samworth, and Stewart \cite{cule2010maximum} to the underlying densities, independently of whether wavelet coefficients across scales exhibit dependencies. We illustrate these findings using the same well-known stochastic models and marathon data. \\ 
Moving beyond log-normality while leveraging log-concavity and finite variance, we revisit a simple estimation method for $c_1$ and $c_2$ based on the central limit theorem, providing confidence intervals and comparing our results with standard bootstrap methods.\\  
Although dependencies between scales are empirically evident, the literature often assumes independence \cite{Morel_etal_2022_dependencies} to facilitate calculations. Adopting this approximation, we conduct a theoretical study of the distribution of wavelet leaders in random wavelet series, which supports the log-concavity model observed numerically. This theoretical and exploratory work represents an initial step toward modeling the distribution of log-leaders, a direction to be further explored in future research.  \\

\noindent
The paper is organized as follows. In Section \ref{Sec: Multifractal analysis}, we define log-leaders and motivate their study by positioning them within the broader context of multifractal analysis, emphasizing its key concepts and results. Section \ref{Sec: Stat model} introduces a new empirically validated statistical model for log-leaders: \textit{log-concave distributions}. We first challenge the common assumption of normality for log-leaders in Section \ref{Subsec: Normality}, demonstrating through both theoretical models and real-world data that this hypothesis must be rejected. We then test the hypothesis of log-concavity for their distributions in Section \ref{Subsec: logconcave}. Assuming log-concavity, we revisit the estimation of $c_1$ and $c_2$ using the central limit theorem and provide confidence intervals in Subsection \ref{Sec: estmation c1 c2}. Section \ref{Sec: theoretical study} presents a brief theoretical analysis of the distribution of log-leaders in random wavelet series, which aligns with the empirical findings from the previous section. Finally, we conclude in Section \ref{Sec: conclusion}, offering insights and suggesting directions for future research.

\section{Multifractal analysis}\label{Sec: Multifractal analysis}

\subsection{Pointwise regularity and multifractal spectrum}

\label{sec:ms}

Let $x\in\R^d\mapsto X(x)$ denote the function or sample path of a stochastic process, or random field, in $d$ variables. 
In order to examine the analytical characteristics of $X$, it is essential to define   a pointwise regularity exponent  of $X$ at each point $x$. The {\em H\"older exponent}  is the most commonly employed notion for pointwise regularity.

\begin{definition}
Let $X: x  \in \R^d \mapsto X(x)$ be a locally bounded
function, $x_0 \in \R^d$, and let $\alpha \geq 0$; $X$ belongs to $C^{\alpha}(x_0)$ if there exist a constant $C > 0$, $r>0$ and a polynomial $P$ of degree less than $\alpha$, such that 
\begin{equation}
   \text{If} \ \vert x-x_0 \vert \leq r, \ \text{then} \ \  \vert X(x) - P(x-x_0) \vert \leq C \vert x-x_0 \vert^{\alpha}.
\end{equation}
The H\"older exponent of $X$ at $x_0$ is 
\begin{equation}
    h_X(x_0)= \sup \left\lbrace  \alpha: X \in C^{\alpha}(x_0) \right\rbrace. 
\end{equation}
\end{definition}
\noindent
Following the initial  intuition of Frisch and Parisi \cite{ParFri85}, a relevant size information on the sets of singularities of $X$ is provided by the following notion.
\begin{definition}
\label{multi-spectr}
Let $X:  \R^d \mapsto \R $ be a locally bounded
function.  The multifractal spectrum  of $X$ is the function $\mathcal{D}_X$ defined on $\R$ by
\begin{equation}
  \forall H\in\R, \qquad \mathcal{D}_X(H)= \dim \left( \left\lbrace  x: h_X(x)=H \right\rbrace\right),
\end{equation}
where 
$\dim$ denotes the Hausdorff dimension 
(with the convention, $\dim (\emptyset) =0$).
\end{definition}
\noindent
In applications, the purpose of multifractal analysis is to obtain numerically robust estimates of the multifractal spectrum. The state-of-the-art methods are based on the wavelet characterization of the H\"older exponent, which we review in the next section.
\subsection{Wavelet based characterization of H\"older exponent}
\subsubsection{Orthonormal wavelet bases}
Let $\{ \psi^{(i)} \}_{i=1,...,2^d-1}$ be  a set of \emph{mother wavelets}, each of them being a smooth   function with fast decay such that the $\psi^{(i)}$ and their derivatives up to a given order $N_\psi$ have fast decay. Furthermore,  the set of dilated (to scales $2^j$) and translated (to space/time location $2^jk$) templates 
\begin{equation}
  \big\{  2^{dj/2} \psi^{(i)}(2^{-j}x-k)  \big\} \quad  \text{for} \  i = 1,...,2^{d}-1 \ , j \in \mathbb{Z},\ \ k \in \mathbb{Z}^{d} ,
\end{equation}
forms an orthonormal basis of $L^2(\R^d)$, see \cite{meyer1992wavelets,Daubechies1988}.  This implies that the mother wavelets $\{\psi^{(i)} \}_{i=1,...,2^d-1}$  possess  vanishing moments: 
$\int_\R P(x) \psi^{(i)}(x) dx = 0 $ for any polynomial $P$ of degree strictly smaller than $N_\psi$. The \emph{wavelet coefficients} of a function $X$ are defined for $i=1,...,2^{d}-1 \ , j \in \mathbb{Z},\ \ k \in \mathbb{Z}^{d}$ as 
\begin{equation}\label{Eq: wavelet coeff}
    c_{j,k}^{(i)}= 2^{dj} \int_{\mathbb{R}^{d}} X(x) \psi^{(i)}(2^{-j}x-k) dx,
\end{equation}
where an $L^1$ normalization is used for wavelet coefficients.
\subsubsection{Wavelet leaders}
Wavelet leaders are multiscale quantities that provide significant advantages over traditional wavelet coefficients in multifractal analysis, both theoretically and practically \cite{Jaffard2004,Jaffard2015}. These properties make them particularly well-suited for developing a multifractal formalism, as will be discussed in Section \ref{spect-estimation}.
\begin{definition}
Let $X$ be a locally bounded function. The wavelet leaders of $X$ are defined as 
\begin{equation}\label{defwl}
\ell_{j,k} = \ell_{\lambda }= \underset{ \lambda ' \subset 3 \lambda}{\sup} \vert c_{\lambda ' }\vert,
\end{equation}
where $k=(k_1,...,k_d)$, $\lambda= \lambda_{j,k}= [k_1 2^{-j} , (k_1 +1)2^{-j} ) \times ... \times [k_d 2^{-j} , (k_d +1)2^{-j} ) $ denotes the dyadic cube of scale $j$, $3 \lambda $ denotes the homothetic cube with same center and three times wider, and $\lambda '$ is a  dyadic cube included in $3 \lambda $  ($\lambda '$ has width $2^{-j'}$ with $j' \geq j$). 
\end{definition}

\noindent
A key property of wavelet leaders is that, under a mild uniform regularity assumption on $X$,  they allow to recover the H\"older exponent in the limit of fine scales through a log-log plot regression, see \cite{Jaffard2004}.

\begin{proposition}
 Let $X: \mathbb{R}^{d} \rightarrow \mathbb{R}$ be such that 
\BE \label{unifreg} \exists \varepsilon >0 \; :  \qquad X \in C^{\varepsilon} (\mathbb{R}^{d}).  \EE
Then, the H\"older exponent of X at $x_0$ can be recovered from the wavelet leaders of $X$  by  
\begin{equation} \label{poiholmq}
    h_X(x_0) = \underset{j \rightarrow + \infty}{\lim  \inf} \dfrac{\log (\ell_{\lambda_{j,k}}(x_0))}{\log ( 2^{-j})},
\end{equation}
where $\lambda_{j,k}(x_0)$ the unique dyadic cube of length $2^{-j}$ which includes $x_0$.
\end{proposition} 
\noindent
When  the relationship \eqref{poiholmq} holds between a  nonnegative quantity $d_{j,k}$ and a pointwise exponent $h$, one says that the $d_{j,k}$ form  a { \em multiresolution quantity associated with the exponent $h$.}
\subsection{Multifractal formalism based on wavelet leaders \label{spect-estimation}}
In practical applications, multifractal analysis provides a global approach to characterizing the fluctuations of $h_X$ across time or space, as captured by the multifractal spectrum. This spectrum can be estimated from data using the mathematical frameworks known as {\em multifractal formalisms}. In this article, we focus on the specific formalism adapted to pointwise H"older regularity, which is based on {\em wavelet leaders}; see \cite{jaffard2007wavelet}.
\begin{definition}
\label{def:scalfun} 
Let $X \in L^{\infty}_{\mathrm{loc}}(\R^d)$, the leader structure fonctions are  defined as the sample moments of the wavelet leaders 
\begin{equation} \label{eqhuit}
    \forall q \in \R, \;\;\;  S_{\ell}(j,q)= 2^{-dj} \sum_{k} \vert \ell_{j,k} \vert^q. 
\end{equation}
The leader scaling function of $X$ is  
\begin{equation}\label{zeta_f}
   \forall q \in \R, \;\;\; \zeta_X(q)= \underset{j \rightarrow + \infty }{\lim \inf} \dfrac{\log(S_{\ell}(j,q))}{\log(2^{-j})}.
\end{equation}
\end{definition}
\noindent
Note that  \eqref{zeta_f} implies that, in the limit of fine scales, $S_{\ell}(j,q)$ exhibits a power law behavior with respect to the scale $2^{-j}$, $ S_{\ell}(j,q) \approx F_q 2^{-j \zeta_X(q)}$. The function $\zeta_X(q)$ is necessarily concave, see \cite{jaffard2007wavelet}; the multifractal formalism is defined via the Legendre transform of the mapping $q \mapsto \zeta_X(q)$. 
\begin{definition}
The leader Legendre spectrum is defined by  
\begin{equation} \label{eq:legtrans}
\forall H\in\R, \qquad    \mathcal{L}_X(H)= \underset{q \in \R}{\inf} ( d + qH- \zeta_{X}(q)).
\end{equation}
\end{definition}
\noindent
If \eqref{unifreg}  holds, then  $\mathcal{L}_X$  provides an upper bound of the multifractal spectrum of $X$ (see \cite{PART2}):
\begin{equation} \forall H \in \R, \qquad 
\mathcal{D}_X(H) \leq \mathcal{L}_X(H).
\end{equation}

\noindent
This result follows from  \eqref{poiholmq}: it holds because  wavelet leaders $\ell_{j,k}$ are multiresolution quantities associated with the exponent $h_X$. Consequently, from a mathematical perspective, the Legendre spectrum based on wavelet leaders provides information about the sizes of singularity sets in the data that would not be accessible if one were to use wavelet coefficients directly (i.e., if the \emph{wavelet scaling function} were employed instead). \\
\noindent
From a numerical standpoint, using wavelet coefficients for $q < 0$ in the definition of the structure function leads to numerical instabilities for the following reason. If $X$ is a random variable, consider its corresponding probability distribution function (PDF):  
\begin{equation} \label{pdfcoeff}  
\text{for } A \geq 0, \qquad F_X (A) = \mathbb{P} ( | X | \leq A ).  
\end{equation}  

\noindent
For many models and in a wide range of applications, the function $F_X$ associated with wavelet coefficients at a given scale is of the form  
\begin{equation}  
F_X (A) \sim c + c' A  
\end{equation}  
(after appropriate rescaling). This is the case, for example, in the widely used generalized Gaussian mixture models; see \cite{Portilla_Strela_W_S_2001,lyu2008modeling,BNHHSJ2}. The case $c \neq 0$ corresponds to mixtures that include a Dirac mass at the origin, meaning that $X$ has no moments of negative order. A linear behavior of $F_X (A)$ for small $A$ is observed, for instance, in distributions with continuous and nonvanishing densities for small $A$; in this case, $X$ has moments of order $q > -1$. Additionally, for wavelet coefficients, the empirical moments
\begin{equation*}
2^{-j} \sum_{i,k}|c^{(i)}_{j,k}|^q
\end{equation*}
which are estimated are numerically highly unstable and are not shift invariant (an appropriate shift of the data can make a wavelet coefficient vanish). However, these moments yield no relevant information on the pointwise regularity of the sample paths because they do not include information on the correlations of the locations of small wavelet coefficients which is crucial for inferring the existence of large values taken by the H\"older exponent (see e.g.\cite{Jaffard2015} where the example of the sample paths of Brownian motion is studied, showing that the Legendre spectrum based on wavelet coefficients does not yield the multifractal spectrum because of the values taken by negative moments). In \cite{BNHHSJ2} a detailed exposition of these questions will be given together with its implications for the multifractal analysis of the sample paths of the corresponding processes. \\
\noindent
The poor behavior of scaling functions based on wavelet coefficients for $q < 0$ contrasts sharply with those based on wavelet leaders. Small values of wavelet leaders indicate that nearby wavelet coefficients also take small values, an event much less likely if the coefficients are independent or weakly correlated. This highlights the  ability of wavelet leaders to capture correlations between the locations of small wavelet coefficients. 
\\
For many signal classes, the PDF \eqref{pdfcoeff} of wavelet leaders is extremely small for small $A$, resulting in distributions with finite moments of all orders. See \cite{Jaf05} for a discussion on how shuffling wavelet coefficients at each scale affects wavelet leader statistics, and \cite{lashermes2008comprehensive} for numerical evidence of significant discrepancies between the statistics of wavelet coefficients and wavelet leaders for small $A$. \\
These theoretical and practical considerations justify the use of both positive and negative values in the structure functions \eqref{eqhuit} as a valuable classification tool. This motivates the study of log-leader distributions in Section \ref{Sec: Stat model} and their relevance for estimating multifractality parameters in Section \ref{Sec: estmation c1 c2}.\\

\noindent Returning to our initial motivation, we now recall how classification parameters are derived from log-leaders. A key observation is that the leader structure functions $S_\ell(j,q)$, defined in \eqref{eqhuit}, can be interpreted as sample mean estimators for the ensemble averages $\mathbb{E} [\vert \ell_{ j,k} \vert ^q]$ (assuming these quantities are finite). According to \eqref{zeta_f}, they exhibit a power-law behavior in the limit of small scales, so that 
\begin{equation} \label{scalings}
    \mathbb{E} [\vert \ell_{ j,k} \vert ^q]  \approx S_\ell(j,q)  \approx  F_q \ 2^{- j \zeta_X(q)}, \quad j \to +\infty.
\end{equation}
Building on the heuristic analysis introduced in \cite{CASTAING1993387} with increments as multiresolution quantities, and later refined for continuous wavelet coefficients in \cite{delour2001intermittency}, we recall the extended interpretation that includes wavelet leaders, see \cite{wendt2008contributions,wendt2007bootstrap,wendt2006bootstrap2,Ciuciu2009} and references therein. Assume that \eqref{scalings} holds 
on a small interval around $q=0$;  taking the logarithm of both sides of \eqref{scalings} yields the   standard cumulant-generating function expansion 
\begin{equation}
\log [ \mathbb{E} ( e^{q \log( \vert \ell_{ j,k} \vert )}) ]  =  \log ( \mathbb{E} [\vert \ell_{ j,k} \vert ^q])  = \log(F_q) + \zeta_X(q) \log(2^{-j}) = \sum_{m \geq 1} C_m(j) \frac{q^m}{m!}, \label{Cum}
\end{equation}
where $C_m(j)$ denotes the $m$-th order cumulant of the random variables $\log(\ell_{ j,k})$.  From \eqref{Cum}, the cumulants $C_m(j)$ must be of the form
\begin{equation}\label{Eq: def cj}
    C_m(j) = c_{0,m} + c_m \ \log(2^{-j}).
\end{equation}
Combining \eqref{Cum} and \eqref{Eq: def cj} gives
\begin{equation}
    \log(F_q) + \zeta_X(q) \log(2^{-j}) = \sum_{m \geq 1} c_{0,m} \frac{q^m}{m!} + \sum_{m \geq 1} c_m \frac{q^m}{m!} \log(2^{-j}).
\end{equation}
Thus, $\zeta_X(q)$ can be expanded around $q=0$ as
\begin{equation}
    \zeta_X(q)= \sum_{m \geq 1} c_m \frac{q^m}{m!}. \label{zeta_expansion}
\end{equation}
Equation \eqref{Eq: def cj} provides a direct method to estimate the coefficients $c_m$, known as log-cumulants, via linear regressions of $C_m(j)$ against $\log(2^{-j})$. Furthermore, assuming $c_2 \neq 0$, the polynomial expansion \eqref{zeta_expansion} can be transformed into an expansion of $\mathcal{L}_X$ around its maximum using the Legendre transform \eqref{eq:legtrans}, see \cite{wendt2008contributions, wendt2009wavelet, Jaffard2015}
\begin{equation}
    \mathcal{L}_X(H)= d+ \frac{c_2}{2!} \left( \frac{H-c_1}{c_2} \right)^2 + \mathrm{remainder}.  \label{estim_L(h)}
\end{equation}
\noindent
Estimating the multifractal spectrum of $X$ entails the theoretical task of evaluating $\zeta_X(p)$ across a broad range of $p$ values and computing its Legendre transform. Equation \eqref{estim_L(h)} provides an alternative approach based on estimating a few parameters $c_p$, which offer a concise synthesis of the multifractal characteristics of the signal $X$. \\
The first log-cumulant $c_1$, for example, indicates the location of the maximum of the multifractal spectrum $\mathcal{D}_X$ and can be interpreted as the almost-everywhere regularity of $X$. Meanwhile, the coefficient $c_2$, often referred to as the \emph{multifractality parameter}, quantifies the width of $\mathcal{D}_X$, reflecting the range of values taken by the H\"older exponent of the signal. However, it is important to note that the expansion in \eqref{zeta_expansion} is not universally valid. Specifically, it requires that the derivatives  of $ \zeta_X$ exist in a neighborhood of $q=0$, a condition that is not satisfied by all processes. 
Consequently, this article focuses on processes for which \eqref{zeta_expansion} remains valid up to the first few orders. Since estimating the Legendre transform $\mathcal{L}_X$ relies on the log-cumulants of the log-leaders, understanding their distribution becomes essential. This is the focus of the next two sections.
\noindent
Estimating the multifractal spectrum of $X$ entails the theoretical task of evaluating $\zeta_X(p)$ across a broad range of $p$ values and computing its Legendre transform. Equation \eqref{estim_L(h)} provides an alternative approach based on estimating a few parameters $c_p$, which offer a concise synthesis of the multifractal characteristics of the signal $X$. \\
The first log-cumulant $c_1$, for example, indicates the location of the maximum of the multifractal spectrum $\mathcal{D}_X$ and can be interpreted as the almost-everywhere regularity of $X$. Meanwhile, the coefficient $c_2$, often referred to as the \emph{multifractality parameter}, quantifies the width of $\mathcal{D}_X$, reflecting the range of values taken by the H\"older exponent of the signal. However, it is important to note that the expansion in \eqref{zeta_expansion} is not universally valid. Specifically, it requires that the derivatives  of $ \zeta_X$ exist in a neighborhood of $q=0$, a condition that is not satisfied by all processes. 
Consequently, this article focuses on processes for which \eqref{zeta_expansion} remains valid up to the first few orders. Since estimating the Legendre transform $\mathcal{L}_X$ relies on the log-cumulants of the log-leaders, understanding their distribution becomes essential. This is the focus of the next two sections.
\section{Statistical model of log-leaders}\label{Sec: Stat model}
As detailed earlier, a key component of multifractal analysis is estimating the multifractal spectrum, which describes the distribution of singularities across scales, achieved numerically by computing wavelet leaders. In this section, we focus on their distribution through the following two questions:\\
\noindent
\textbf{\underline{[1] Questioning the normality of log-leaders}}
Previous works have often assumed that log-leaders follow a Gaussian distribution \cite{wendt2008contributions,combrexelle2015bayesian,Combrexelle2016EUSIPCO,combrexelle2016multifractal,wendt2013bayesian},
based on analyses using QQ-plots, which graphically compare the quantiles of two distributions (see Section \ref{Qplot}). These plots suggest a visual agreement between the empirical quantiles of log-leaders and those of a standard normal distribution. However, this method has only been experimentally applied to certain self-similar processes and multifractal multiplicative cascades \cite{combrexelle2016multifractal}, leaving the normality assumption unverified for other stochastic processes and possibly unreliable for real-world data. In Subsection \ref{Subsec: Normality}, we aim to test this assumption in different contexts, showing that despite a good visual fit, classical normality tests do not support this assumption for many theoretical and real-world data sets.  \\
\noindent
\textbf{\underline{[2] A new model for the law of log-leaders: log-concave distributions}} Our second aim is to present a more general methodology by moving beyond the restrictive normality assumption and focusing, in Section \ref{Subsec: logconcave}, on a more flexible hypothesis: log-concavity of the underlying distributions. Log-concavity is a powerful statistical assumption that encompasses normality while applying to a broader range of distributions. This added flexibility allows to model the complexity of multifractal data without imposing excessive constraints.
Our methodology involves using a statistical procedure to test whether log-leader distributions are log-concave, without assuming a specific form. This approach provides a more robust characterization of the underlying structures in the studied signals. It represents a significant advancement by moving beyond the limiting assumption of normality, thereby paving the way for more nuanced analyses suitable for the diverse behaviors observed in fields such as finance, biology, and other applications where multifractality prevails.
\subsection{Checking the normality assumption of log-leaders}\label{Subsec: Normality}
As mentioned in the introduction, it is commonly assumed in the multifractal analysis literature that log-leaders are normally distributed (e.g. \cite{wendt2008contributions,combrexelle2015bayesian,Combrexelle2016EUSIPCO,combrexelle2016multifractal,wendt2013bayesian}).
The objective of this section is to challenge this hypothesis by applying a normality test to the log-leaders calculated from simulated processes (fractional Brownian motion, compound Poisson motion...) and real physiological data from marathon runners. Most of the numerical results are obtained using a Daubechies mother wavelet $\psi$ with $N_\psi$ = 3 vanishing moments. At the end of Subsection \ref{Subsec: logconcave test}, to discuss the robustness of our results with respect to a change of wavelet basis, we compare them in the specific case of a multifractal random walk using two different Daubechies mother wavelets, $\psi$, with $N_\psi=1$ and $N_\psi=4$.
\color{black}
\subsubsection{Graphical method: Quantile-quantile plots}
\label{Qplot}
The normal quantile-quantile plot (Q-Q plot) is the most commonly used and effective diagnostic tool for checking the normality of data. It enables a graphical evaluation of the conformity between two distributions by plotting their quantiles against each other. Figure \ref{qqplot} illustrates Q-Q plots depicting the quantiles of the empirical distributions of log-leaders at different scales $j=4,5,6$ associated with various processes and real signals, against the quantiles of the standard normal distribution. Below is a concise overview and a representative selection of the processes and real data that will be used for conducting the tests.
\begin{enumerate}
\label{processes}
    \item \textbf{Fractional Brownian motion (fBm):} It is the only Gaussian self-similar process with stationary increments and its multifractal properties are controlled by one single parameter,  the \emph{Hurst exponent} $H$, which takes values in $(0,1)$, see \cite{flandrin1992wavelet}. For the tests we set $H = 0.4$ and $H = 0.7$.
    \item \textbf{Canonical Mandelbrot Cascade motion (CMC-motion):} It is defined as $A_r(t) = \underset{r \rightarrow 0}{\lim} \int_0^t \mathcal{Q}_r(s)ds$, where $\mathcal{Q}_r$ is the canonical Mandelbrot cascade (CMC) that represents a archetype of multifractal measure, see \cite{Heurteaux2016}. It is  defined at resolution $r = 2^{-J}$ (i.e., after $J$ iterations) as: 
    $ \mathcal{Q}_r(t) = c \displaystyle \prod_{(j,k)\in\mathcal J}  W_{jk}$, where \[ \mathcal J=\{(j,k),\,j=1,\dots,J \quad\text{ and }\,k\, \mbox{ such that } \quad t \in [2^{-j}k, 2^{-j}(k+1)] \}, \] the $W_{jk}$ are the multipliers for the $j$-th iteration and $k=1,...,2^j$. The multipliers have to be strictly positive random variables, and satisfy the constraint $\mathbb{E}[W]= 1$, ensuring that the cascade conserves mass in average. In this work, we will only consider CMC with log-Normal multipliers $W = 2^{-U}$, where $U \sim \mathcal{N}(\mu,\sigma^2)$ is a random variable following a Gaussian distribution of mean $\mu$ and variance $\sigma^2$. Conservation of mass implies $ \sigma^2 = 2 \mu / \ln(2)$, see \cite{wendt2008contributions}. For the numerical simulation, we use $\mu=0.37$.
    \item \textbf{Compound Poisson motion (CPM):} It is defined as $A_r(t) = \underset{r \rightarrow 0}{\lim} \int_0^t \mathcal{Q}_r(s)ds$, where $\mathcal{Q}_r$ is the compound Poisson cascade (CPC), which is a multifractal measure based on multiplicative constructions, see \cite{mandelbrot1999intermittent}. 
    The CPC signal is defined as: $ \mathcal{Q}_r(t) = c \displaystyle \prod_{(t_i,r_i)\in C_t(t)}  W_i$, where $(t_i,r_i)$ are the positions of a two-dimensional Poisson point process with intensity measure $dm(t,r)$ supported on the rectangle $I=\{(t ,r) : r_{\min} \leq r \leq  1 , -1/2 \leq t \leq T + 1/2 \}$, with $r_{\min}\in(0,1]$ and $T>0$ being fixed; the multipliers $W_i$ are defined below. We define the cone of influence at point $t$ as 
    $C_r(t)=\{(t',r') : r \leq r' \leq 1,t - r'/2 \leq t' \leq t+r'/2\}. $ The normalizing constant $c$ is defined such that $\mathbb E[\mathcal{Q}_r(t)] = 1$.
    For more details, see \cite{barral2000continuity,barral2002multifractal,chainais2002compound,chainais2007infinitely}. We use here CPCs with two different types of multipliers $W_i$:
    \begin{itemize}
        \item First, the compound Poisson cascades with log-normal (CPC-LN) multipliers of the form $W= \exp(Y)$, where $Y \sim \mathcal{N}(\mu,\sigma^2)$.
        \item Second, the compound Poisson cascades with log-Poisson (CPC-LP) multipliers where the multipliers $W$ are reduced to a constant $w$. 
    \end{itemize}
The full algorithmic statement of the construction procedure of CPC's signals can be found in \cite{frezza2019simulating} and a visual representation is shown in \cite{wendt2008contributions}. Results are obtained using $T=100$, $r_{\min}=0.02$, $\sigma^2=0.2$, $\mu=- \sigma^2 /2$ and $w=3/2$. 
\item \textbf{Multifractal random walk (MRW)}: It stands as another notable member within the category of multifractal multiplicative cascade based processes. MRW is a non Gaussian process with stationary increments, which is defined as $X(t) = \sum_{t} G_H(t) \exp(W(t))$
where the $G_H(t)$ are the increments of a fractional Brownian motion with Hurst parameter $H$ and $W$ is a Gaussian process which is independent of $G_H$ and with covariance $\mathrm{cov}[( W(t_1), W(t_2) ]= \beta^2 \log \left( \dfrac{L}{\Vert t_1-t_2\Vert+1} \right)$. For more details, see \cite{bacry2001multifractal,abry2009multifractal}. In application we set $H=0.6$, $N=L=10^5$ and  $\beta = 0.05$.

\item \textbf{Real data:} 
We use two types of physiological data from 20 amateur marathon runners: heart rate (measured in beats per minute) and speed (in m/s) of sample size ranging between $11000$ and $12000$ samples. These data were collected between 2023 and 2024 from different 42 km marathons using the Garmin Forerunner 630 watch and were graciously provided by V. Billat and her team. Examples of these two types of data from a marathon runner are shown in Fig. \ref{data}. We picked this particular real-life example because it is a typical case where multifractal analysis techniques have proved relevant in order to supply new insight on the physiological mechanisms that produced these data, see the articles of V. Billat and her collaborators \cite{wesfreid05a,billat2009detection,JSW2}.
For multifractal analysis, we made two choices based on the nature of the data. In the case of cardiac data, the scales were selected between $j = 5$ and $j = 7$ (i.e., between 26 s and 3 min 25 s), as they have been identified as relevant for such physiological signals \cite{abry2010methodology}. Conversely, no universally recognized a priori scales exist for speed data. To address this, we determined common scales across the datasets that provided the best log-log regressions, ultimately selecting a range between $j = 5$ and $j = 8$ (i.e., between 32 s and 4 min 27 s).  
\begin{figure}[!htbp]
\centering
 \includegraphics[width=0.9\textwidth]{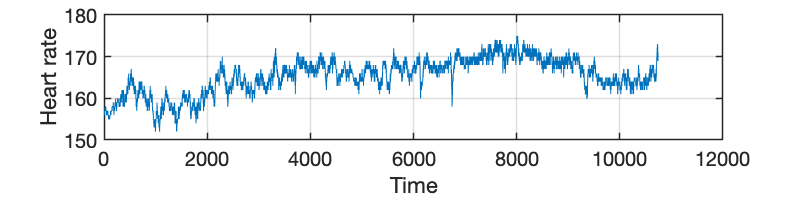}
  \vfill
  \includegraphics[width=0.9\textwidth]{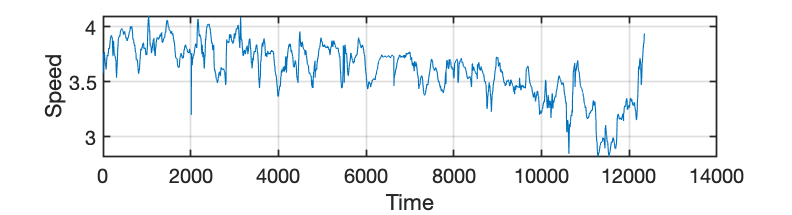}
\caption{Representation of two physiological data heart rate (beats per minute) and speed (kilometer per hour) of a marathon runner respectively of sample size 11000 and 12000. }
\label{data}
\end{figure}
\end{enumerate}
\noindent
These QQ-plots suggest that the normal distribution offers a reasonable approximation for the marginal distribution of log-leaders of these processes and these example of real data; this contrasts sharply with   leaders, as their marginal distributions are found to significantly  deviate from Gaussian distributions. Although the approximation may seem appealing, normality is a rather strong assumption, and it is therefore important to  
 complement graphical methods with formal numerical techniques and statistical normality tests  in order to robustly  assess or dismiss. The formal methods which will be tried in the following section offer a more rigorous and quantitative evaluation of whether the data follow a normal distribution.

\begin{figure}[!htbp]
\centering
 \includegraphics[scale=0.5]{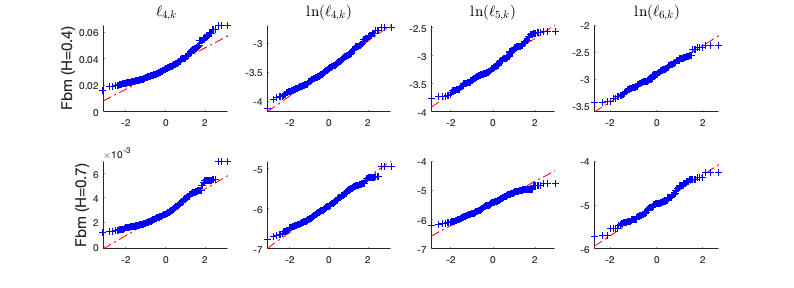}
\vfill
\includegraphics[scale=0.5]{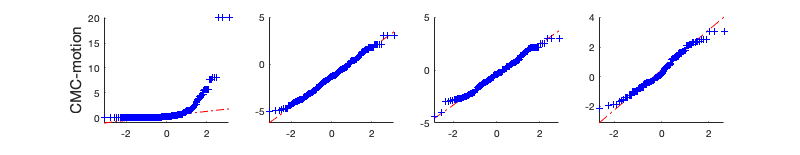}
\vfill
\includegraphics[scale=0.5]{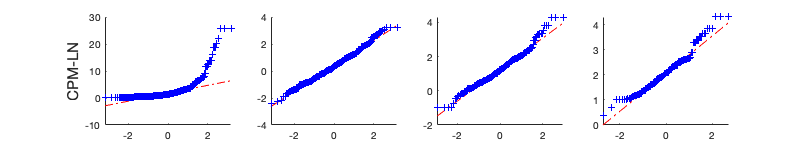}
\vfill
    \includegraphics[scale=0.5]{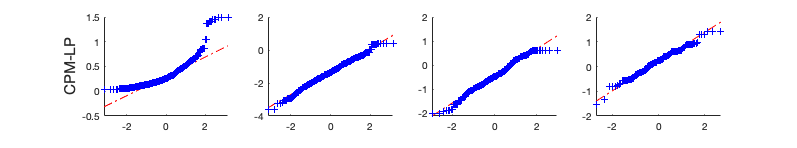}
\vfill
    \includegraphics[scale=0.5]{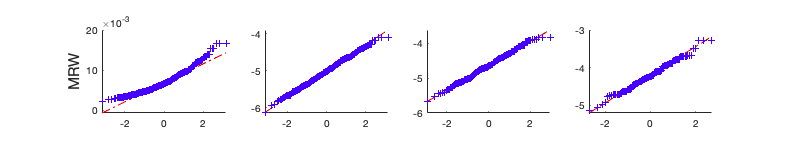}
    \vfill
    \includegraphics[scale=0.5]{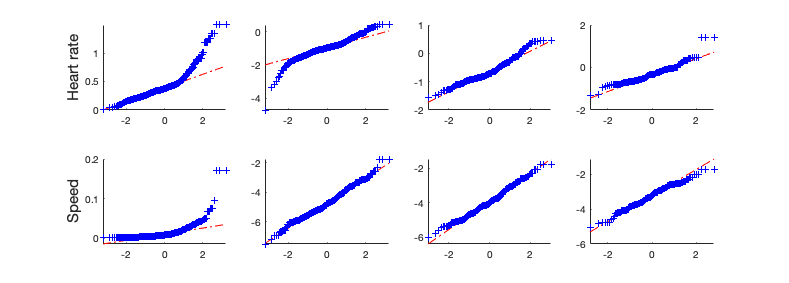}
\caption{Quantile-quantile plots of the empirical distributions (\textcolor{blue}{+}) of leaders $\ell_{j,k}$ and the log-leaders $\log(\ell_{j,k})$ at several scales $j=4,5,6$ against standard normal distribution (\textcolor{red}{-.-}) for a selection of processes of a sample size equal to $10^5$, uniformly sampled over the time interval $[0,1]$ described in section \ref{processes};  from top to bottom: fBm of  Hurst exponents 0.4 and 0.7, Canonical Mandelbrot Cascade motion, two types of  compound Poisson motions, multifractal random walk and finally,  heart rate and speed data respectively of sample size around 11000 and 12000 collected on marathon runner.
}
\label{qqplot}
\end{figure}
\subsubsection{Numerical method: Shapiro-Wilk test }\label{Subsec: Shapiro test}

\underline{\textbf{Testing normality}} The most common normality test procedures available in statistical toolbox are the Shapiro-Wilk test, Kolmogorov-Smirnov test, Anderson-Darling test and Lilliefors test. Studies by Pala (2003) \cite{mendes2003type}, Keskin (2006) \cite{keskin2006comparison}, and Razali (2011) \cite{article} have shown that the Shapiro-Wilk test is the most powerful normality test. It was the first test able of detecting deviations from normality caused by skewness, kurtosis, or both. Initially limited to sample sizes between 3 and 50, improvements to its algorithm have extended its applicability to data sets of up to 5000 observations (see
\cite{althouse1998detecting,royston1982extension,royston1982algorithm,royston1995remark}). The null hypothesis of this test is that the population is normally distributed, whereas the alternative hypothesis is that it is not.
Thus, if the $p$-value is less than the chosen significance $\alpha$-level, then the null hypothesis is rejected and there is evidence that the data tested are not normally distributed. On the other hand, if the $p$-value is greater than the chosen $\alpha$-level, then the null hypothesis, i.e. the data came from a normally distributed population, cannot be rejected. \vspace{3pt}\\
\underline{\textbf{Experiments}} 
We run 1,000 simulations of various processes in MATLAB, as described above. For each simulation, we test the normality of the log-leaders $ \log(\ell_{j,k})$, where $j=4,5,6$ using the Shapiro-Wilk test, as detailed in Subsection \ref{Subsec: Shapiro test}. For each process, we report in Table \ref{Shapiro-Wilk_1} the proportion of times that the null hypothesis was rejected with $\alpha=0.05$. In addition, we use physiological data (heart rate and speed) from 20 marathon runners such that for each signal we perform a normality test on log-leaders and subsequently we report in Table \ref{Shapiro-Wilk_2} the proportion of times where the null hypothesis was rejected.
Tables \ref{Shapiro-Wilk_1} and \ref{Shapiro-Wilk_2} clearly demonstrate a significant proportion of rejections of the normality hypothesis. We observe that as $j$ decreases, the proportion of rejections increases, indicating a stronger rejection of the normality of the log-leaders. This phenomenon can be explained by the toolbox we utilize to perform multifractal analysis and obtain the log-leaders: as $j$ decreases, we analyze smaller scales, resulting in larger size of log-leader vector at scale $j$. Consequently, the statistical test becomes more precise, yielding more significant results.

\begin{table}[h!]
\caption{Empirical rejection probabilities of the Shapiro-Wilk test (see Subsection \ref{Subsec: Shapiro test}) applied to the log-leaders $ \log(\ell_{j,k})$, $j=4,5,6$ of 1000 simulations of different known processes.}
\label{Shapiro-Wilk_1}
\begin{tabular}{@{}llccc@{}}
\toprule
$X$ &  & $\log(\ell_{4,k})$  & $\log(\ell_{5,k})$ & $\log(\ell_{6,k})$ \\
\midrule
\multirow{2}{*}{fBm} & $H=0.4$ & 0.953 & 0.968 & 0.816 \\
                     & $H=0.7$ & 0.842 & 0.681 & 0.624 \\
\midrule
CMC-motion &  & 0.783  & 0.723  & 0.636 \\
CPM-LN     &  & 0.867  & 0.817  & 0.793 \\
CPM-LP     &  & 0.990  & 0.942  & 0.767 \\
MRW        &  & 0.8670 & 0.6980 & 0.6250 \\
\botrule
\end{tabular}
\end{table}

\begin{table}[h!]
\caption{Empirical rejection probabilities of the Shapiro-Wilk test (see Subsection \ref{Subsec: Shapiro test}) applied to the log-leaders $ \log(\ell_{j,k})$, $j=4,5,6$ of physiological data (heart rate and speed) from 20 marathon runners.}\label{Shapiro-Wilk_2}%
\begin{tabular}{@{}llll@{}}
\toprule
{$X$}  & $\log(\ell_{4,k})$ & $\log(\ell_{5,k})$ & $\log(\ell_{6,k})$ \\
\midrule
{Heart rate} &  1    &       1        &      0.75      \\
{Speed} &   1     &       1        &      0.9     \\
\botrule
\end{tabular}
\end{table}

\subsection{Assessing the log-concavity of  log-leaders}\label{Subsec: logconcave}

As empirically demonstrated in the previous section, the restrictive assumption of normality for log-leaders is frequently contradicted by standard normality tests, whether applied to synthetic samples from classical mono- and multi-fractal processes or to real-life data. In this section, we introduce a non-parametric model for the density of log-leaders by considering log-concave densities. One motivation for this checking the relevance of this broader framework is that it is nonetheless relevant to carry out statistical methods in order to obtain confidence intervals, see Section \ref{Sec: estmation c1 c2} where these methods are mentioned.  
\subsubsection{Log-concave densities}
The class of log-densities gathers densities on $\mathbb R^d$ which are log-concave according to the following definition.
\begin{definition}
A function $f:\R^d\rightarrow[0,\infty)$ is \emph{log-concave} if $\log(f)$ is concave, with the convention that $\log 0=-\infty$.
\end{definition}
\noindent
Denote by $ \mathscr F_0$ the set of upper-continuous log-concave densities on $\R^d$. 
This expansive class includes many commonly encountered parametric families. Examples of univariate log-concave densities include Gaussian densities, Gumbel densities, logistic densities, $\Gamma(\alpha,\lambda)$ densities with $\alpha \geq 1$, Beta$(a,b)$ densities with $a,b\geq 1$ and Laplace densities. Multivariate Gaussian densities are also log-concave, as are uniform densities on convex, compact sets.
Moreover, this class possesses several closure and stability properties, making it a very natural infinite-dimensional generalization of the class of Gaussian densities. Among its remarkable properties, let us mention that $\mathscr F_0$ is closed under linear transformations, marginalisation, conditioning and convolution \cite{ibragimov_1956_LC_composition,walther_2009_inference}.\vspace{3pt}\\
\underline{\textbf{Applications}}
Log-concave densities naturally arise in various applications, making them highly relevant for modeling. In practice, assuming that the sample distribution is log-concave rather than normal can offer significant advantages. For example, in finance, stock returns often exhibit heavy tails and asymmetry that are inadequately captured by a normal distribution (see \cite{mandelbrot_1997_variation}). Cryptocurrencies, with their high volatility, display even more pronounced heavy-tail returns, which are better modeled by distributions derived from the generalized gamma distribution \cite{Van_2022_novel}. Other applications in finance, insurance, and operations research are explored in \cite{Resnick_2007_heavy} and \cite{Badia_2021_log}. In image processing, denoising methods assuming Gaussian noise may fail to accurately model real-world scenarios involving impulse noise; log-concave distributions like the exponential distribution better capture outliers and sudden pixel value changes \cite{Nandal_2018_sensitivity}. Additionally, log-concave distributions simplify analysis in reliability \cite{sengupta_1999_reliability} and survival theory \cite{Bagnoli_2006_log}, where the cumulative distribution functions (CDFs) often lack closed forms. The log-concavity of a density function implies that of its CDF, which aids statistical analysis in nonparametric frameworks \cite{Bagnoli_2006_log}.
\color{black}
\vspace{3pt}\\
\underline{\textbf{Estimation and testing for log-concave densities}} Despite their widespread use in diverse statistical applications due to their adaptability and computational convenience, estimating log-concave densities in high-dimensional spaces presents significant challenges. Cule, Samworth, and Stewart \cite{cule2010maximum} introduced an innovative method for estimating multidimensional log-concave densities using maximum likelihood estimation (MLE). They proved the existence and uniqueness of the log-concave MLE (MLE-LC) estimator $\widehat{f}_n$ and described an algorithm for its computation, which is implemented in the R package LogConcDEAD. The theoretical properties of the MLE-LC estimator were further explored in a companion paper by Cule and Samworth \cite{Cule_Samworth_2010}.\\
In the next subsection, we recall the definition of the log-concave maximum likelihood estimator introduced in \cite{cule2010maximum}, and describe the log-concavity test that follows. We apply this test to different processes and real-world datasets to examine the log-concavity of the associated log-leaders.
\color{black}

\subsubsection{Log-Concave Maximum Likelihood Estimator (MLE-LC)}

Let $X_1, \dots, X_n$ be independent and identically distributed (i.i.d.) random variables with density $f_0$ on $\mathbb{R}^d$ that we want to estimate, with $n \geq d+1$.  The non-parametric MLE, denoted as $\widehat{f}_n^{\mathrm{MLE-LC}}$, is defined as
\begin{equation}\label{Eq: MLE-LC def}
\widehat f_n=\widehat f_n^{\mathrm{MLE-LC}}:=\argmax_{f\in\sF_{\mathrm{0}}}\prod_{i=1}^{n}f(X_i)=\argmax_{f\in\sF_{\mathrm{0}}}\log\Big(\prod_{i=1}^{n}f(X_i)\Big)
\end{equation}
where, as a reminder, $\sF_{0}$ denotes the class of log-concave densities on $\mathbb{R}^d$. \\
Figure \ref{fn_chap} presents three examples of $\log(\widehat f_n)$  computed from the random variables $\log(\ell_{4,k})$ derived from two physiological data heart rate and speed, as well as the multifractal random walk process (MRW). It can be clearly observed that the maximum likelihood estimator $\widehat f_n$ is indeed log-concave.
\noindent
The statistical properties of the estimator are studied in Cule and Samworth (2010) \cite{Cule_Samworth_2010}. Importantly, their key finding does not rely on the assumption that the underlying density $f_0$ is log-concave. This is particularly noteworthy since determining the log-concavity of $f_0$ from a data sample is practically impossible. Thus, it becomes imperative to ensure that the estimator behaves reasonably even when this condition is not met.\\
\noindent
We first recall that the Kullback-Leibler 
divergence of a density $f$ from $f_0$ is given by
\begin{equation*}
d_{\mathrm{KL}}(f_0,f) =  \displaystyle \int_{\mathbb{R}^d} f_0 \log\Big(\dfrac{f_0}{f}\Big)d\mu,
\end{equation*}
where $\mu$ is the Lebesgue measure on $\R^d$.
Let $E$ be the support of $f_0$, i.e. the smallest closed set such that $\int_{E} f_0d\mu =1$, denoting $\mathrm{int}(E)$ the interior of $E$ (the largest open set contained in $E$) and $\log_+(x) = \max \lbrace  \log(x) ,0 \rbrace$. The following theorem  provides a key result to assess the properties of the log-concave maximum likelihood estimator.
\begin{figure}[!htbp]
\centering
 \includegraphics[width=0.3\textwidth]{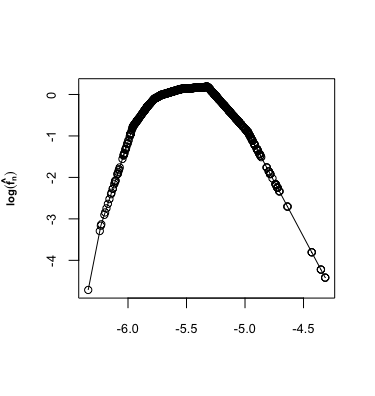}
  \includegraphics[width=0.3\textwidth]{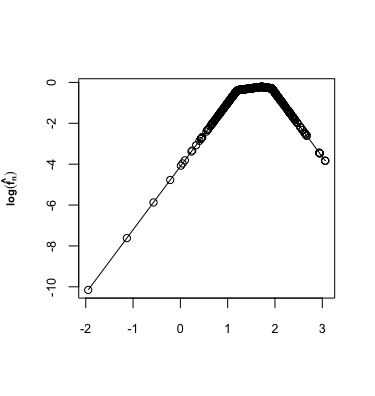}
  \includegraphics[width=0.3\textwidth]{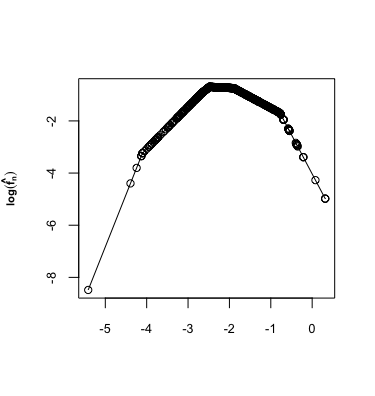}
\caption{$\log(\widehat f_n)$ for $\log(\ell_{4,k})$ of multifractal random walk (Left), heart rate (Middle) and speed (Right).}
\label{fn_chap}
\end{figure}
\begin{theorem}\label{theo}(\cite{Cule_Samworth_2010}, Theorem 4)
Let $f_0$ be any density on $\mathbb{R}^d$ satisfying \[ \int_{\mathbb{R}^d} \Vert x \Vert f_0(x) dx  < \infty,  \quad  \int_{\mathbb{R}^d} f_0 \log_+(f_0)d\mu  < \infty \quad \text{and}\quad \mathrm{int}(E)\neq \emptyset. \] There exists a log-concave density $f^*$, unique almost everywhere, such that
\begin{equation}\label{Eq: MLE-LC fstar}
f^*=\argmin_{f\in\mathscr{F}_0}d_{\mathrm{KL}}(f_0,f),
\end{equation}
i.e., that minimizes the Kullback-Leibler divergence of $f$ from $f_0$ over all log-concave densities $f$. \\
Let $a_0>0$ and $b_0 \in \mathbb{R} $ such that $ f^*(x) \leq \exp(-a_0 \Vert x \Vert +b_0)$; then,  for any $a<a_0$, 
\begin{equation}\label{Eq: MLCE-LC weights}
\displaystyle \int_{\mathbb{R}^d} \exp( a \Vert x \Vert) \  \vert \widehat{f}_n(x) - f^*(x) \vert dx  \rightarrow 0   \ \ \ almost \  surely \ as \ n \rightarrow \infty,
\end{equation}
and, if $f^*$ is continuous, then 
\begin{equation}\label{Eq: MLCE-LC consistent}
\underset{x \in \mathbb{R}^d}{\sup} \  \exp( a \Vert x \Vert) \  \vert \widehat{f}_n(x) - f^*(x) \vert dx  \rightarrow 0 \quad \mbox{  almost surely as } \quad n \rightarrow \infty.
\end{equation} 
\end{theorem}
\noindent
We can interpret the results of this theorem as follows: when $f_0$ is log-concave, \eqref{Eq: MLCE-LC weights} means that the log-concave maximum likelihood estimator $\widehat{f}_n$ is strongly consistent within specific exponentially weighted total variation metrics, and the convergence is ensured also if $f_0=f^*$ is continuous (see \eqref{Eq: MLCE-LC consistent}). The convergence in these exponentially weighted norms provides confidence in the estimator's performance in the extreme tails of the density.
When $f_0$ is not log-concave,  their exists a unique log-concave density $f^*\neq f_0$ defined by \eqref{Eq: MLE-LC fstar}, i.e. that minimizes the Kullback-Leibler divergence from $f_0$. Furthermore, the log-concave maximum likelihood estimator $\widehat{f}_n$ converges in the same manner to $f^*$.
\subsubsection{The log-concavity test}
In  \cite{cule2010maximum}, Cule, Samworth, and Stewart investigate the test \
\begin{equation*}
\cH_0: f \;\text{is log-concave }\quad \text{vs} \quad \cH_1: f \;\text{is not log-concave.}
\end{equation*}
To address this general problem, they devise a permutation test based on the log-concave maximum likelihood estimator described above. We briefly describe its principle before outlining the steps of the corresponding algorithm.
\\
Consider $n$ random variables $X_1, \dots, X_n$ i.i.d. with density $f_0$ that we want to test for log-concavity. The random sample is denoted $\mathcal{X} = (X_1, \dots, X_n)$.\\
\noindent
\underline{\textbf{Test Principle}}: The test is based on the Log-Concave Maximum Likelihood Estimator (MLE-LC) defined by \eqref{Eq: MLE-LC def}. The underlying idea is the following:
\begin{itemize}
\item If $f_0$ is truly log-concave, i.e. \textbf{$\mathcal{H}_0$ true}, then $\widehat{f}_n$ applied to the data $X_1, \dots, X_n$ provides a good estimation of $f_0$. The test is interpreted as follows: if $d(f_0, \widehat{f}_n)$ is \textit{small}, \textbf{we keep $\mathcal{H}_0$.}
\item If $f_0$ is not log-concave, i.e. \textbf{$\mathcal{H}_1$ true}, then $\widehat{f}_n$ applied to the data $X_1, \dots, X_n$ converges to some $f^*\neq f_0$ defined by \eqref{Eq: MLE-LC fstar} and does not provide a good estimation of $f_0$. \
The test is interpreted as follows: if $d(f_0, \widehat{f}_n)$ is \textit{large}, \textbf{we reject $\mathcal{H}_0$}.
\end{itemize}
\noindent
\underline{\textbf{Test Algorithm}}:
\begin{enumerate}
\item \textbf{Fit the MLE-LC} to the data/sample $\mathcal{X} = (X_1, \dots, X_n)$ by using the function \textcolor{blue}{\textbf{mlelcd}} from the package \textcolor{blue}{\textbf{LogConcDEAD}}. This provides a probability density $\widehat{f}_n$.
\item \textbf{Generate} a sample $\mathcal{X}^{\star} = X_1^\star, \dots, X_n^\star$ from $\widehat{f}_n$ by using the function \textcolor{blue}{\textbf{rlcd}} from the package \textcolor{blue}{\textbf{LogConcDEAD}}. This provides $n$ values $x_1^\star, \dots, x_n^\star$, which are realizations of $n$ random variables $X_1^\star, \dots, X_n^\star$ with density $\widehat{f}_n$.
\item \textbf{Compare} the samples $\mathcal{X} = (X_1, \dots, X_n)$ and $\mathcal{X}^\star = (X_1^\star, \dots, X_n^\star)$ by means of the distance
\begin{equation*}
T=\sup_{A\in\cA_0}|P_n(A)-P_n^\star(A)|,
\end{equation*}
where $\mathcal{A}_0$ is the set of balls centered at a point in $\mathcal{X} \cup \mathcal{X}^\star$ and $P_n$ is the empirical distribution associated with $\mathcal{X}$

\begin{equation*}
P_n(A)=\frac{1}{n}\sum_{i=1}^n\delta_{X_i}(A).
\end{equation*}

In the one-dimensional case  with realizations $x_1, \dots, x_n$ of $X_1, \dots, X_n$ and $x_1^\star, \dots, x_n^\star$ of $X_1^\star, \dots, X_n^\star$, we consider the test statistic 

\begin{equation}\label{Eq_Stat_test}
T=\sup_{r>0}\sup_{x\in\{x_1,\dots,x_n,x_1^\star,\dots,x_n^\star\}}\bigg|\frac{1}{n}\sum_{i=1}^n\delta_{X_i}\big((x-r,x+r)\big)-\delta_{X_i^\star}\big((x-r,x+r)\big)\bigg|.
\end{equation}

\item \textbf{"Shuffle the stars"} to increase the test power. Choose a permutation $\pi$ (uniform random) of ${1, \dots, 2n}$ to mix the elements of $\mathcal{X} \cup \mathcal{X}^\star$ and place stars on the last $n$ elements of the obtained sample
\begin{equation*}
x_1,\dots,x_n,x_1^\star,\dots,x_n^\star\xrightarrow{\pi} x_{\pi(1)},\dots,x_{\pi(n)},x_{\pi(n+1)},\dots,x_{\pi(2n)}=:x_{1,1},\dots,x_{n,1},x_{1,1}^\star,\dots,x_{n,1}^\star.
\end{equation*}
Compute as above
\begin{equation*}
T_1^\star=\sup_{A\in\cA_0}|P_{n,1}(A)-P_{n,1}^\star(A)|,
\end{equation*}
where
\begin{equation*}
P_{n,1}(A)=\frac{1}{n}\sum_{i=1}^n\delta_{X_{i,1}}(A)\quad \text{and}\quad P_{n,1}^\star(A)=\frac{1}{n}\sum_{i=1}^n\delta_{X_{i,1}^\star}(A).
\end{equation*}
\item \textbf{Repeat the procedure $B-1$ times}: We obtain $T_1^\star, \dots, T_B^\star$. Then, we consider the corresponding order statistics  $T_{(1)}^\star \leq \dots \leq T_{(B)}^\star$.
\item For a level of significance $\alpha$, we \textbf{reject $\mathcal{H}_0$ if $T > T^\star_{(B+1)(1-\alpha)}$}. The test is given by:
\begin{itemize}
\item the test statistic $T$ defined by \eqref{Eq_Stat_test},
\item the rejection region is $T^\star_{(B+1)(1-\alpha)}$.
\end{itemize}
\end{enumerate}

\subsubsection{Efficiency of the log-concavity test}
To illustrate the performance of the log-concavity test, we compare its results with theoretical expectations using distributions with known properties. Note that while mixtures of log-concave densities can sometimes be log-concave, they generally are not. We use a mixture of standard univariate normal densities $f:x\mapsto\frac{1}{2} \phi(x)+\frac{1}{2} \phi(x- c)$ where $\phi$ denotes the standard normal density and $\Vert c \Vert \in \lbrace 0,2,4 \rbrace$. This mixture is log-concave if and only if $\Vert c \Vert \leq 2$. We also examine several parametric families of univariate distributions with log-concave densities, including normal, Gumbel, Laplace, logistic, gamma, and Weibull distributions with a shape parameter of at least one. Additionally, we test the generalized Gaussian distribution, which is log-concave if and only if the shape parameter $\beta $ satisfies $\beta >1$. For comparison, we include non-log-concave densities such as Cauchy, Pareto, and lognormal. The log-concavity test was performed on these distributions with a sample size of $n=200$ and for the log-normal density we used different sizes $n=200, 500$ and $1000$. The test statistic is the proportion of rejections out of 100 repetitions, with $B=99$ and the null hypothesis was rejected with a level of significance $\alpha=0.05$. The results are summarized in Table \ref{test2} below.\\

\begin{table}[h!]
\centering
\caption{\label{Tab: log-oncave}Empirical rejection probabilities of the log-concavity test (see Subsection \ref{Subsec: logconcave}) applied to the usual distribution densities. The sample size is $n=200$. The rejection proportions lead to the decision (third column) to either retain (LC) the null hypothesis $\mathcal{H}_0$ that the density is log-concave, to reject it (non-LC), or to neither reject nor retain $\mathcal{H}_0$ (inconclusive). The last column indicates whether (based on known theoretical results) the density is truly log-concave or not.}
\label{test2}
\begin{tabular}{@{}llccc@{}}
\toprule
\multicolumn{2}{c}{\textbf{Density}} 
& \textbf{Prop. of reject $\mathcal{H}_0$} 
& \textbf{Decision} 
& \textbf{Theoretical result} \\
\midrule

\multirow{3}{*}{$f(x) = \dfrac{1}{2} \phi(x) + \dfrac{1}{2} \phi(x - c)$} 
  & $c = 0$ & 0.03  & LC          & LC \\
  & $c = 2$ & 0.02  & LC          & LC \\
  & $c = 4$ & 0.29  & Inconclusive & Non LC \\
\midrule

\multicolumn{2}{l}{Gumbel: $f(x,1,1) = e^{e^{1 - x}}$} 
& 0.02  & LC  & LC \\
\midrule

\multicolumn{2}{l}{Laplace (0,1): $f(x) = \dfrac{1}{2} e^{-\vert x \vert}$} 
& 0     & LC  & LC \\
\midrule

\multicolumn{2}{l}{Logistic: $f(x,0,1) = \dfrac{e^{-x}}{(1 + e^{-x})^2}$} 
& 0     & LC  & LC \\
\midrule

\multicolumn{2}{l}{Gamma: $f(x,1,2) = \dfrac{1}{2\Gamma(1)} e^{-x/2}$} 
& 0.01  & LC  & LC \\
\midrule

\multicolumn{2}{l}{Weibull: $f(x,2,2) = \dfrac{x}{2} e^{-x^2 / 4}$} 
& 0.01  & LC  & LC \\
\midrule

\multicolumn{2}{l}{Cauchy: $f(x,0,1) = \dfrac{1}{\pi} \dfrac{1}{x^2 + 1}$} 
& 1     & Non LC & Non LC \\
\midrule

\multicolumn{2}{l}{Pareto: $f(x,1,1) = \dfrac{1}{x^2}$} 
& 1     & Non LC & Non LC \\
\midrule

\multirow{2}{*}{Generalized Gaussian: $f(x) = \dfrac{\beta e^{-\vert x \vert^\beta}}{2 \Gamma(1/\beta)}$}
  & $\beta = 0.5$  & 0.93 & Non LC & Non LC \\
  & $\beta = 1.5$  & 0    & LC      & LC \\
\midrule

\multirow{2}{*}{Lognormal: $f(x) = \dfrac{1}{x \sqrt{2\pi}} e^{- \dfrac{\log^2(x)}{2\sigma^2}}$}
  & $n = 200$  & 0.06  & \textcolor{red}{\textbf{LC}   }      & \textcolor{red}{\textbf{Non LC} }\\
  & $n = 1000$ & 0.37  & Inconclusive & Non LC \vspace{2.2mm}\\
\bottomrule
\end{tabular}
\vspace{2mm}
\textbf{\underline{Table legend}}: LC: log-concave, Non LC: non log-concave.
\end{table}
\noindent
In light of the results reported in Table \ref{Tab: log-oncave}, it is clear that the test is effective and, in most cases, the outcomes align  with the nature of each distribution. 
When the decisions to reject or not reject the null hypothesis $\mathcal H_0$ do not align with whether the underlying density is truly log-concave or not, it is worth noting that the test can still detect a significant departure from log-concavity. In this case, the rejection rate is too high ($>\!>5\%$) for us to retain $\mathcal{H}_0$ with confidence, which leads to the decision to neither reject nor retain $\mathcal{H}_0$ (indicated in the table as \emph{inconclusive}). This occurs in the log-concave case with $n=1000$, where a larger sample allows for contradicting the incorrect 'LC' decision obtained with $n=200$ test runs. Similarly, in the case of a mixture with $c=4$, the high rejection rate suggests running the test on a larger sample (see experiments with different sample sizes in Cule and Samworth \cite{Cule_Samworth_2010}) to refine the result and ultimately reject $\mathcal{H}_0$.
\subsubsection{Application of  the log-concavity test to log-leaders}\label{Subsec: logconcave test}

\begin{table}[h!]
\centering
\caption{Log-concavity test on log-leaders $\log(\ell_{j,k})$, where $j \in \{ 4, 5, 6 \}$, for different processes and real data. The test result shows the proportion of times out of 100 repetitions that the null hypothesis $\mathcal{H}_0$ was rejected.}
\label{test}
\begin{tabular}{@{}llccc@{}}
\toprule
\multicolumn{2}{c}{\textbf{$X$}} 
& $\log(\ell_{4,k})$ 
& $\log(\ell_{5,k})$ 
& $\log(\ell_{6,k})$ \\
\midrule

\multirow{2}{*}{fBm} 
  & $H = 0.4$ & 0.01 & 0.18 & 0.02 \\
  & $H = 0.7$ & 0.00 & 0.13 & 0.02 \\
\midrule

\multicolumn{2}{l}{CMC-motion}  & 0.01 & 0.02 & 0.02 \\
\midrule
\multicolumn{2}{l}{CPM-LN}      & 0.05 & 0.05 & 0.02 \\
\midrule
\multicolumn{2}{l}{CPM-LP}      & 0.00 & 0.00 & 0.00 \\
\midrule
\multicolumn{2}{l}{MRW}         & 0.02 & 0.04 & 0.03 \\
\midrule

\multicolumn{2}{l}{Heart rate}  & 0.08 & 0.04 & 0.04 \\
\midrule
\multicolumn{2}{l}{Speed}       & 0.17 & 0.06 & 0.02 \\
\bottomrule
\end{tabular}
\end{table}
\noindent
We use the same processes and real data as in Section \ref{Qplot} to perform the log-concavity test on log-leaders. The results in Table \ref{test} clearly show that the proportion of times the null hypothesis was rejected out of 100 repetitions is very low. This confirms that the log-concavity of the distribution of log-leaders is a realistic hypothesis. This is an encouraging result, especially when working with real data where no theoretical properties are known.
Given the practical implications of this finding, it would be valuable to validate this hypothesis across a broader range of real life datasets.
Testing the log-concavity of log-leaders on different types of  data, drawn from various fields such as finance, biology, or physics, would not only strengthen the generalizability of this hypothesis but also reveal its potential limitations.
\\
A relevant question to address is the effect of changing the mother wavelet $\psi$ on the results of the log-concavity test of the log-leaders. 
To explore this, using the example of the Multifractal Random Walk (with the same parameters as in Section \ref{Qplot}), we choose to apply the test with two wavelet bases that differ significantly in both regularity and support size, specifically Daubechies wavelets with varying numbers of vanishing moments: $N_\psi = 1$ corresponding to the Haar wavelet, and $N_\psi = 4$. Notably, a smaller $N_\psi$ corresponds to a more irregular wavelet with smaller support. The results of the log-concavity test on the log-leaders, shown in Table \ref{vanishing}, indicate that the choice of wavelet does not affect the test outcomes. Indeed, the proportion of times the null hypothesis was rejected across 100 repetitions remains consistently low, reinforcing the potential validity of the log-concavity hypothesis for the distribution of log-leaders. This example indicates that the test results are robust to changes in the choice of the wavelet basis. In other words, the log-leaders appear to maintain their log-concavity regardless of the wavelet basis's characteristics, such as its regularity and support. However, while this result of robustness is promising, it is only based on numerical evidence and  a  thorough theoretical study would be  needed to support these findings in all generality.

\begin{table}[!htbp]
\centering
\caption{Log-concavity test on log-leaders $\log(\ell_{j,k})$, where $j \in \{4, 5, 6\}$, for a multifractal random walk (MRW) using Daubechies wavelets with different vanishing moments $N_\psi = 1$ and $N_\psi = 4$. The test result shows the proportion of times out of 100 repetitions that the null hypothesis $\mathcal{H}_0$ was rejected.}
\label{vanishing}
\begin{tabular}{@{}llccc@{}}
\toprule
\multicolumn{2}{c}{\textbf{Wavelet}} 
& $\log(\ell_{4,k})$ 
& $\log(\ell_{5,k})$ 
& $\log(\ell_{6,k})$ \\
\midrule

\multicolumn{2}{l}{$N_\psi = 1$} & 0.01 & 0.00 & 0.03 \\
\midrule
\multicolumn{2}{l}{$N_\psi = 4$} & 0.01 & 0.12 & 0.00 \\
\bottomrule
\end{tabular}
\end{table}
\section{Estimation of multifractality parameters}\label{Sec: estmation c1 c2}
\noindent
In this section, we focus on estimating two important multifractality parameters: $c_1$ and  $c_2$ (we leave the statistical estimation of the third major multifractality parameter  and $H^{min}_f$ for the forthcoming article \cite{BNHHSJ2}). An estimator for the coefficients $c_1$ and $c_2$ was introduced by Jacyna et al. \cite{jacyna_2023_improved} and comes with theoretical guarantees through confidence intervals. We examine the assumptions underlying its applicability by leveraging the analysis of log-leader distributions presented in Section \ref{Sec: Stat model}.
\subsection{Estimation of $c_1$ and $c_2$}\label{Subsec: estimation c1 c2}
In this section, we revisit the estimation procedure for $c_1$ and $c_2$ developed in Jacyna et al. \cite{jacyna_2023_improved}, questioning the assumptions of its application. Consider a process $X=(X(t))_{t\in\R_+}\in L_{\mathrm{loc}}^{\infty}(\R^d)$ whose wavelet coefficients are defined by \eqref{Eq: wavelet coeff} and wavelet leaders by \eqref{defwl}.
Recalling \eqref{Eq: def cj} the coefficients $c_1$ and $c_2$ are defined by the equations
\begin{equation}\label{Eq: def cj2}
C_1(j) = c_{0,1} + c_1 \log(2^{-j})\quad\text{and}\quad C_2(j) = c_{0,2} + c_2 \log(2^{-j}). 
\end{equation}
where for $m\in\{1,2\}$, $C_m(j) $ denotes the $m$-th order cumulant of the random variables $\log(\ell_{ j,k})$. Consider $N$ i.i.d. copies $X^1,\dots,X^N$ of the process $X$ and $M$ scales $j_1,\dots,j_M$. This leads to define the estimates $(\widehat c_{0,1}^{(i)},\widehat c_{1}^{(i)})$ and $(\widehat c_{0,2}^{(i)},\widehat c_{2}^{(i)})$ by the equations
\begin{align}\label{Eq: hat c1 hat c2}
\widehat c_{0,1}^{(i)} + \widehat c_{1}^{(i)} \log(2^{-j}) 
&= \frac{1}{n_j}\sum_{k=1}^{n_j} \log(\ell_{j,k}^{(i)}) 
=:\widehat\mu_j^{(i)} \nonumber \\
\text{and} \quad 
\widehat c_{0,2}^{(i)} + \widehat c_{2}^{(i)} \log(2^{-j}) 
&= \frac{1}{n_j} \sum_{k=1}^{n_j} \left( \log(\ell_{j,k}^{(i)}) - \widehat\mu_j^{(i)} \right)^2 
=:\widehat\sigma_j^{(i)}.
\end{align}

where, for any $j\in\{j_1,\dots,j_M\}$, $n_j$ stands for the number of translations $k$ at scale $j$. For any $i\in\{1,\dots,N\}$, we define the vectors $\bsb{\widehat{\mu}}^{(i)}=(\widehat{\mu}_{j_1}^{(i)},\dots,\widehat{\mu}_{j_M}^{(i)})$ and $\bsb{\widehat{\sigma}}^{(i)}=(\widehat{\sigma}_{j_1}^{(i)},\dots,\widehat{\sigma}_{j_M}^{(i)})$. 
The least-squares estimates $\bsb{\widehat c}_1^{(i)}:=(\widehat c_{0,1}^{(i)},\widehat c_{1}^{(i)})$ and $\bsb{\widehat c}_2^{(i)}:=(\widehat c_{0,2}^{(i)},\widehat c_{2}^{(i)})$ can be rewritten as
\begin{equation}\label{Eq: estimators c1i c2i}
\bsb{\widehat c}_1^{(i)}=(\widehat c_{0,1}^{(i)},\widehat c_{1}^{(i)}):=(\mathbf{H}^\top\mathbf{H})^{-1}\mathbf{H}^\top\widehat{\bsb{\mu}}^{(i)} \quad\text{and}\quad \bsb{\widehat c}_2^{(i)}=(\widehat c_{0,2}^{(i)},\widehat c_{2}^{(i)}):=(\mathbf{H}^\top\mathbf{H})^{-1}\mathbf{H}^\top\widehat{\bsb{\sigma}}^{(i)}, 
\end{equation}
where the matrix $\bf{H}$ is given by
\begin{equation*}
\bf{H}=
\begin{pmatrix}
1&\log(2^{-j_1})\\
1&\log(2^{-j_2})\\
\vdots&\vdots\\
1&\log(2^{-j_M})\\
\end{pmatrix}.
\end{equation*}
We consider then the estimators
\begin{align}\label{Eq: estimators c1 c2}
\bsb{\widehat c}_1 
&:= (\widehat c_{0,1},\widehat c_{1}) 
:= \frac{1}{N}\sum_{i=1}^N\bsb{\widehat c}_1^{(i)} 
= \frac{1}{N}\sum_{i=1}^N(\mathbf{H}^\top\mathbf{H})^{-1}(\mathbf{H}^\top \widehat{\bsb{\mu}}^{(i)}) \nonumber \\
\text{and} \quad
\bsb{\widehat c}_2 
&:= (\widehat c_{0,2},\widehat c_{2}) 
:= \frac{1}{N-1}\sum_{i=1}^N\bsb{\widehat c}_2^{(i)} 
= \frac{1}{N-1}\sum_{i=1}^N(\mathbf{H}^\top\mathbf{H})^{-1}\mathbf{H}^\top\widehat{\bsb{\sigma}}^{(i)}.
\end{align}

By definition, $C_1(j) = \mathbb{E}[\log(\ell_{j,k})] =: m_j$ for all $j \in \{j_1, \dots, j_M\}$, and letting $\mathbf{m} = (m_1, \dots, m_M)$, it follows that $\mathbb{E}[\bsb{\widehat{c}}_1] = (\mathbf{H}^\top \mathbf{H})^{-1} \mathbf{H}^\top \mathbf{m} = (c_{0,1},c_1)=:\bsb{c}_1$. Therefore, $\bsb{\widehat{c}}_1$ is an unbiased estimator of $\bsb{c}_1:=(c_1,\dots,c_1)\in\R^N$. A similar argument shows that $\bsb{\widehat{c}}_2$ is an unbiased estimator of $\bsb{c}_2$. \\
Expanding \eqref{Eq: estimators c1i c2i}, we get
\begin{equation}\label{Eq: expansion c1}
\widehat c_1^{(i)}=\sum_{k=1}^{M}\big(\bsb{h}_{2,1}+\bsb{h}_{2,2}\log(2^{-j_k})\big)\widehat \mu_{j_k}^{(i)}
\end{equation}
by writing $(\mathbf{H^\top H})^{-1}=(\bsb{h}_{i,j})_{(i,j)\in\{1,2\}^2}$. 
\\
Assuming that the log-leader distributions are log-concave (an hypothesis empirically validated in Section \ref{Sec: Stat model}), the random variables $\log(\ell_{j,k})$ have finite moments of all orders, including a finite variance. By the Central Limit Theorem (CLT),  
\begin{equation}\label{Eq: TCL}
Z_N:=\sqrt{N}\frac{(\widehat{c}_1 - c_1) }{\sigma_{c_1}}\xrightarrow[N \to \infty]{\mathcal{L}} \mathcal{N}(0, 1),
\end{equation}
where $\sigma_{c_1}$ is the common standard deviation of the $\widehat c_1^{(i)}$.
From this, we derive the confidence interval (valid for $N \geqslant 30$):
\begin{equation*}
c_1 \in \left[\widehat{c}_1 - z_{\alpha/2} \frac{\widehat{\sigma}_{{c}_1}}{\sqrt{N}}, \, \widehat{c}_1 + z_{\alpha/2} \frac{\widehat{\sigma}_{{c}_1}}{\sqrt{N}}\right],
\end{equation*}
where $\widehat{\sigma}_{{c}_1}$ is the empirical estimator of $\sigma_{c_1}$  and $z_{\alpha/2}$ is the $(1 - \alpha/2)$-quantile of the Normal distribution, corresponding to the desired confidence level $1 - \alpha$. The quality of the normal approximation provided by the Central Limit Theorem \eqref{Eq: TCL} can be assessed using the Berry-Esseen theorem, which states that
\begin{equation*}
\sup_{x\in\R}\big|\mathbb P(Z_N\leq x)-\Phi(x)\big|\leq \frac{C}{c_2^{3/2}N^{3/2}}\sum_{i=1}^N\mathbb E\big[|\widehat c_1^{(i)}-c_1|^3\big],
\end{equation*}
where $C$ is some constant smaller than $0.46$ (see \cite{shevtsova_2014_berry}). Note that by \eqref{Eq: expansion c1} we get
\begin{equation*}
\mathbb E\big[|\widehat c_1^{(i)}-c_1|^3\big]=\sum_{k=1}^{M}\big(\bsb{h}_{2,1}+\bsb{h}_{2,2}\log(2^{-j_k})\big)\mathbb E[|\widehat \mu_{j_k}^{(i)}-m_k|^3]<\infty,
\end{equation*}
since the third absolute moment is always finite due to the exponential (or faster) decay of the tails of log-concave distributions. We thus recover the estimation results and a confidence interval for the estimates of $c_1$ and $c_2$ stated in Jacyna et al. \cite{jacyna_2023_improved}, without relying on the assumption that the log-leaders form a stationary Gaussian field, as assumed in their work and refuted in Subsection \ref{Subsec: Shapiro test}.
\subsection{Numerical experiments}
\noindent
\textbf{\underline{Procedure:}}
We consider two types of processes with different multifractal properties (described in Section \ref{Qplot}). The first is fractional Brownian motion (fBm), a monofractal process, where the parameter $c_1$ coincides with the Hurst exponent $H$, and $c_2 = 0$. The second is the Multifractal Random Walk (MRW), a multifractal process, for which $c_1$ is given by $c_1 = H + \beta^2 / 2$ and $c_2 = -\beta^2$. For the numerical study, we simulate $N = 100$ realizations of each stochastic process, with a sample size of $1310$, uniformly sampled over the time interval $[0,1]$. Estimation is performed using Daubechies wavelets with $N_{\psi} = 3$ vanishing moments, and the choice of scales $(j_1, \dots, j_M)$ is described below.  \\
\noindent
\textbf{\underline{Tuning Parameter}}  
To describe the scale selection procedure, we consider $N = 100$ simulations of fBm with $H = 0.8$. Performing a multifractal analysis for each process, we examine the log-log regression used to estimate the parameter $h_{\min}$, which corresponds to the uniform regularity of the function in the scale provided by global H\"older spaces (see Figure \ref{scale}). Numerically, it is determined as the slope of the linear regression between the scales and the logarithm of the supremum of the wavelet coefficients at each scale. This choice of regression is justified as it provides a robust estimate of this classification parameter (see \cite{Jaffard2015}).  
\\
\noindent
We extract the range of scales that yields the best linear regression fit, based on the coefficient of determination, denoted $R^2$. The coefficient $R^2$ measures the goodness of fit of the regression, with values closer to 1 indicating a stronger linear relationship between the variables. It is defined as
\[
R^2 = 1 - \frac{\sum_{i} (y_i - \widehat{y}_i)^2}{\sum_{i} (y_i - \bar{y})^2},
\]
where  $y_i$ represents the observed values, $\widehat{y}_i$ the predicted values from the regression model, and $\bar{y}$ the mean of the observed values. After $N $ simulations, we select the most frequently occurring scale range $\{ j_1, \dots, j_M\}$ and use it for the estimation of the multifractal parameters $c_1 $ and $c_2 $, see Figure \ref{scale}. 
\begin{figure}[!htbp]
\centering
 \includegraphics[scale=0.3]{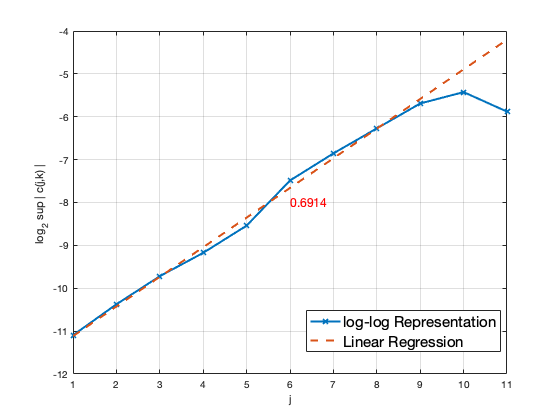}
  \hfill
  \includegraphics[scale=0.3]{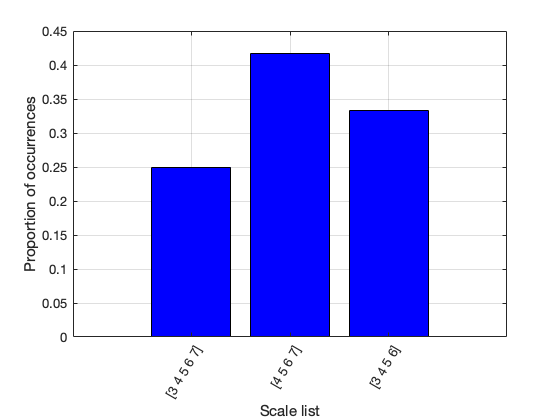}
\caption{Left, log-log plot illustrating the linear regression performed to estimate $h_{min}$ for a fractional Brownian motion with Hurst parameter $H = 0.8$. Right, histogram displaying the proportion of occurrences of the different scale intervals selected as yielding the optimal linear regression across $N = 100$ independent realizations.}
\label{scale}
\end{figure}
\\
\noindent
\textbf{\underline{Results and interpretation:}}
In Table \ref{IC}, we present the confidence intervals for the parameter $ c_1$ for different realizations of fBm and MRW, choosing $ H \in \{ 0.5, 0.6, 0.7, 0.8, 0.9 \}$ and setting $\beta = 0.1 $. In each case, we compare the confidence intervals ($95 \% $ confidence level, $\alpha=0.05$) obtained with those estimated using the percentile bootstrap method, see \cite{wendt2008contributions, wendt2007bootstrap, wendt2006bootstrap,wendt2006bootstrap2}. The percentile bootstrap method is a non-parametric approach for constructing confidence intervals without assuming a specific distribution for the data. It relies on resampling with replacement from the original dataset to generate $B$ bootstrap samples. For each sample, the parameter of interest is estimated, denoted by $\widehat{\theta}_b$. The empirical distribution of these bootstrap estimates is then used to determine the $\alpha$ quantiles $\widehat{\theta}_\alpha$ corresponding to the desired confidence level. Thus, the confidence interval at level $1-\alpha$ is 
\[
\text{CI}_{\theta, (1 - \alpha)}^{\text{per}} = \left[ \widehat{\theta}_{\frac{\alpha}{2}}, \widehat{\theta}_{1 - \frac{\alpha}{2}} \right].
\]
Thus, for the bootstrap method, we adopt a sample size of $131000$, ensuring that any differences in confidence interval widths are attributable to the intrinsic properties of the estimation techniques rather than disparities in sample size. Thus, the comparison of confidence intervals highlights that those obtained using our CLT-based estimation method generally have smaller widths than the bootstrap-derived intervals. This reduction in interval width suggests that our approach provides a more precise estimation of the parameter $c_1$, thereby enhancing statistical efficiency. While the bootstrap method tends to yield slightly wider intervals, reflecting a more conservative quantification of uncertainty, the difference remains moderate, ensuring that both methods offer consistent inference. The narrower confidence intervals obtained via the CLT-based approach indicate its potential advantage in terms of precision, while the bootstrap method remains a robust alternative, particularly useful in scenarios where parametric assumptions may not hold. These findings demonstrate the effectiveness of our method in providing refined estimates while maintaining coherence with an established resampling technique.
\noindent
Table \ref{IC-c2} presents the estimates of the $c_2 $ coefficient fBm and MRW models across different values of $H \in \{0.5, 0.6, 0.7, 0.8, 0.9\}$ and fixed $\beta = 0.1$. The theoretical value of $ c_2 $ is $0$ for fBm and $-0.01$ for MRW. Three estimation methods are compared: our  proposed CLT-based method $\widehat{c}_2$, the percentile bootstrap estimator $\widehat{c}_2^b$, and the Bayesian estimator $\widehat{c}_2^{\mathrm{BY}}$ which is based on the assumption of the normality of log-leaders, see \cite{combrexelle2015bayesian,Combrexelle2016EUSIPCO,wendt2013bayesian}. For both the fBm and MRW models, the CLT-based estimator consistently provides the most reliable estimates, yielding results that are closest to the theoretical values. This method demonstrates stability and precision, with minimal deviations across all values of the Hurst exponent. In contrast, the bootstrap estimator exhibits higher variability, leading to less accurate estimates. For instance, in the MRW model with $H = 0.7 $, the bootstrap estimator produces a value of  $\widehat{c}_2^b = -0.039 $, which deviates significantly from the theoretical value of $-0.01$. 
Meanwhile, the Bayesian estimator provides reasonably accurate results for MRW, but its estimates exhibit more fluctuation compared to the CLT-based method. 
These findings highlight the CLT-based method as the most accurate and stable for estimating $c_2$, consistently yielding values closest to the truth, especially for MRW. In contrast, the bootstrap and Bayesian methods show greater variability and slightly lower precision.
\\
\textbf{\underline{Conclusion:}} 
The CLT-based estimator offers comparable precision to other methods for estimating $c_1$ and $c_2$ but stands out for its weaker assumption-requiring only log-concavity of the log-leaders. Additionally, it is significantly faster, with a runtime of approximately 14 seconds, comparable to the Bayesian method, and nearly five times shorter than the bootstrap method, which takes about 1 minute and 10 seconds.
\begin{table}[h!]
\centering
\caption{Confidence intervals for the parameter $c_1$ for fractional Brownian motion (fBm) and Multifractal Random Walk (MRW) with the same sample size, considering different values of Hurst parameter $H \in \{0.5, 0.6, 0.7, 0.8, 0.9\}$ and fixed $\beta = 0.1$. Two estimation methods are compared: our proposed CLT-based method (CLT method) and the percentile bootstrap method (Bootstrap) with $B=100$ bootstrap samples, both at $95\%$ confidence level ($\alpha = 0.05$). The columns \textbf{LB} and \textbf{UB} denote the lower and upper bounds of the confidence intervals, respectively, while \textbf{UB-LB} represents the interval width.}
\label{IC}

\begin{tabular}{@{}lccc|ccc@{}}
\toprule
\multicolumn{7}{c}{\textbf{fBm}} \\
\midrule
Theoretical $c_1$ 
& \multicolumn{3}{c|}{CLT method} 
& \multicolumn{3}{c}{Bootstrap} \\
\cmidrule(r){2-4} \cmidrule(l){5-7}
  & LB & UB & UB-LB 
  & LB & UB & UB-LB \\
\midrule
0.5 & 0.4709 & 0.5071 & \textbf{0.0362} & 0.485  & 0.532  & 0.047 \\
0.6 & 0.5811 & 0.6076 & 0.0265          & 0.584  & 0.606  & \textbf{0.022} \\
0.7 & 0.6845 & 0.7161 & \textbf{0.0316} & 0.6808 & 0.7152 & 0.0344 \\
0.8 & 0.7842 & 0.8126 & \textbf{0.0284} & 0.773  & 0.823  & 0.05 \\
0.9 & 0.89   & 0.9197 & \textbf{0.0297} & 0.87   & 0.931  & 0.061 \\
\bottomrule
\end{tabular}

\vspace{0.5cm} 

\begin{tabular}{@{}lccc|ccc@{}}
\toprule
\multicolumn{7}{c}{\textbf{MRW}} \\
\midrule
Theoretical $c_1$ 
& \multicolumn{3}{c|}{CLT method} 
& \multicolumn{3}{c}{Bootstrap} \\
\cmidrule(r){2-4} \cmidrule(l){5-7}
  & LB & UB & UB-LB 
  & LB & UB & UB-LB \\
\midrule
0.505 & 0.4641 & 0.5052 & \textbf{0.0411} & 0.464 & 0.508 & 0.044 \\
0.605 & 0.5784 & 0.6055 & \textbf{0.0271} & 0.585 & 0.617 & 0.032 \\
0.705 & 0.681  & 0.7102 & \textbf{0.0292} & 0.696 & 0.734 & 0.038 \\
0.805 & 0.7938 & 0.8292 & \textbf{0.0354} & 0.8   & 0.85  & 0.05 \\
0.905 & 0.8693 & 0.9224 & \textbf{0.0531} & 0.88  & 0.937 & 0.057 \\
\bottomrule
\end{tabular}
\end{table}

\begin{table}[h!]
\caption{Estimates of the $c_2$ parameter for fractional Brownian motion (fBm) model ($c_2=0$) and the multifractal random walk (MRW) model ($c_2=-0.01$), for different values of Hurst exponent $H \in \{0.5, 0.6, 0.7, 0.8, 0.9\}$ and fixing $\beta = 0.1$. The estimates are obtained using our proposed CLT-based method ($\widehat{c}_2$), the percentile bootstrap method ($\widehat{c}_2^b$) (B=100), and the Bayesian method ($\widehat{c}_2^{\mathrm{BY}}$). Bold values indicate the negative estimates closest to the theoretical values.}
\label{IC-c2}
\begin{tabular*}{\textwidth}{@{\extracolsep\fill}cccc|cccc@{}}
\toprule
\multicolumn{4}{c|}{\textbf{fBm \hspace*{0.3cm} $c_2=0$}} & \multicolumn{4}{c}{\textbf{MRW \hspace*{0.3cm} $c_2=-0.01$}} \\
\midrule
H & $\widehat{c}_2$ & $\widehat{c}_2^b$ & $\widehat{c}_2^{\mathrm{BY}}$ & H & $\widehat{c}_2$ & $\widehat{c}_2^b$ & $\widehat{c}_2^{\mathrm{BY}}$ \\
\midrule
0.5 & -0.0118 & 0.008  & \textbf{-0.0036} & 0.5 & \textbf{-0.01} & -0.018  & -0.0127 \\
0.6 & \textbf{-0.0001} & -0.004  & -0.0051 &  0.6 & \textbf{-0.0084} & -0.017 & -0.0152 \\
0.7 & \textbf{-0.004} & -0.004  & -0.0049 & 0.7 & \textbf{-0.0113} & -0.039 & -0.0159 \\
0.8 & \textbf{-0.002} & 0.002 & -0.0066 & 0.8 & \textbf{-0.0155} & -0.018  & -0.0172 \\
0.9 & \textbf{-0.001} & 0.011 & -0.0074 & 0.9 & \textbf{-0.0104} & -0.006   & -0.0183 \\
\bottomrule
\end{tabular*}
\end{table}

\section{ Theoretical study of the  distribution of log-leaders for   random wavelet series} 
\label{Sec: theoretical study}
Previous statistical tests (see Section \ref{Sec: Stat model}) provided a broad understanding of the distribution of log-leaders: they are not normally distributed, but their distribution appears to be log-concave. To illustrate these experimental findings, we provide a first attempt to study the distribution of the log-leaders associated to {\em random wavelet series} $X=(X(t))_{t\in\R_+}$ defined by
\begin{equation}\label{Eq: decomp wavelet process}
X(t)= \sum_{j,k} 2^{-\alpha j}  \ X_{j,k} \ \psi(2^{j}t-k),
\end{equation}
where $\alpha\in\mathbb R_+^*$, the $2^{j/2} \psi(2^{j}\cdot+k) $ ($j,k \in \ZZ$) form an orthonormal wavelet basis and the wavelet coefficients satisfy the following assumptions (see \cite{AJ02} and therein): 
\begin{itemize}
\item The wavelet coefficients are independent both \emph{within} and \emph{across} scales,
\item The random variables $X_{j,k}$ appearing in \eqref{Eq: decomp wavelet process} are generalized Gaussians, i.e. their common density is defined on $\R$ by
\begin{equation}\label{Eq: density generalized Gaussian}
f_\beta(x)= \dfrac{\beta}{2 \Gamma (\frac{1}{\beta})} \e^{- \vert x \vert^\beta}=: \kappa_\beta \e^{-\vert x \vert^\beta},
\end{equation}
 where $\Gamma$ denotes the gamma function, $\beta >0$ is the shape parameter and $\kappa_\beta>0$ is a normalizing constant which belongs to the interval $(0,0.565)$ for any $\beta>0$.
 The family of densities $\{f_\beta,\,\beta>0\}$ includes the normal distribution (with variance $1/2$) when 
$ \beta =2$ and the Laplace distribution when $\beta=1$.
\end{itemize}
\noindent
We justify the assumption that the variables are generalized Gaussians rather than standard Gaussians based on considerations from signal and image processing \cite{wainwright2001random}, where such model assumptions are derived form the inspection of large collections of data. In these fields, wavelet coefficients often exhibit non-Gaussian distributions: most are small, but a few have large magnitudes, leading to heavy-tailed distributions when $\alpha<2$ \cite{lyu2008modeling}. Bayesian estimation methods have integrated non-Gaussian priors using Gaussian Mixture Models (GMM) \cite{Chipman_Kolaczyk_McCulloch_1997}, a mixture of a Gaussian and a point mass at $0$ \cite{abramovich1998wavelet}, and Gaussian scale mixtures for image denoising \cite{Portilla_Strela_W_S_2001, Strela_Portilla_Simoncelli_2000}.

\begin{remark}
To streamline the computations, we assume in this section that the wavelet coefficients are independent both \emph{across} and \emph{within} scales, meaning that the $X_{j,k}$ are independent in both $j$ and $k$. This brief analysis aims to theoretically validate the log-concavity model of log-leaders proposed in this paper, while recognizing that it relies on the widely used yet restrictive assumption of independence between scale coefficients.
We also justify neglecting these dependencies because we propose to provide an approximation rather than an exact form of the log-leaders' law. The long-term objective is to explore the implications that scale dependencies have on the laws of log-leaders; part of this program is developed in the forthcoming paper \cite{BNHHSJ2}. 
\end{remark}
\color{black}
The previous  assumptions entail a simplified model of a continuous law exhibiting non-zero density in the vicinity of 0. 
We identify the pair $(j,k)$ and $\lambda:=\lambda_{j,k}:= [ k/2^{j}, (k+1)/2^{j}]$ and we define the leaders as 
\begin{equation}\label{defwl2}
    \ell^1_{j,k} = \ell^1_{\lambda }= \underset{ \lambda ' \subset  \lambda}{\sup} \vert c_{\lambda ' }\vert.
\end{equation}
Note that we do not use $\lambda ' \subset 3 \lambda$ to guarantee independence between leaders at the same scale, despite its advantages in facilitating theoretical results and improving numerical efficiency. 
However, the standard  proof establishing that the Legendre spectrum provides an upper bound for the multifractal spectrum, see \cite{Jaffard2004}, relies on a definition of the scaling function based on 3-leaders, as given in \eqref{defwl}, i.e., with the quantity
\[  \ell^3_{j,k} = \ell^3_{\lambda }= \underset{ \lambda ' \subset  3 \lambda}{\sup} \vert c_{\lambda ' }\vert ,    \]
see 
\eqref{zeta_f}; in order to validate our simplifying choice here, 
it is therefore important to verify that the definitions of  scaling functions using  either $\ell^1_{\lambda }$ or $\ell^3_{\lambda }$  actually coincide. 
Let us introduce the following notations: 
we define the corresponding wavelet leader structure functions   \[ \mbox{ for } i\in\{1, 3\}, \quad  j\geq 0, \quad \mbox{ and } 
    \forall q \in \R, \qquad  S_i(j,q)= 2^{-dj} \sum_{k} \vert \ell^i_{j,k} \vert^q, \]
and 
the leader scaling functions  as 
\[  \mbox{ for } i\in\{ 1, 3\}, \quad 
   \forall q \in \R, \qquad  \zeta_i(q)= \underset{j \rightarrow + \infty }{\lim \inf} \dfrac{\log(S_i(j,q))}{\log(2^{-j})}.
\] 

\begin{proposition} \label{equalscal} 
  For any locally bounded function, the two  leader scaling functions defined above coincide, i.e. 
\begin{equation} \label{equalzeta} 
\forall q \in \R, \qquad  \zeta_1(q)=\zeta_3(q). 
\end{equation}
\end{proposition}
\begin{proof}
This result is proved in  Appendix \ref{App: equality leaders}.    
\end{proof}

\begin{remark}
At a given scale $j$, all leaders $ \ell^1_{j,k} $ share the same distribution, and the distributions across scales can be derived from $ \ell^1_{0,0} $ by rescaling with the factor $2^{-\alpha j}$. Hence, it is sufficient to determine the distribution of $\ell^1_{0,0}$. We denote $I := \lambda_{0,0} = [0,1]$. 
\end{remark}
\noindent
Under previous assumptions, we are led to compute
\begin{align}\label{Eq: distribution leaders product}
\nonumber
\mathbb{P}(\ell^1_{0,0} \leq A) 
&= \mathbb{P}\Big( \underset{ \lambda \subset  I }{\sup} \vert c_{\lambda} \vert \leq A \Big) \\
\nonumber
&= \underset{ (j,k):\lambda\subset I}{\prod} \mathbb{P}(\vert c_{j,k} \vert \leq A )\\
&= \underset{(j,k):\lambda\subset I}{\prod} \mathbb{P}( \vert  X_{j,k}\vert \leq 2^{\alpha j} A )\\
&=\underset{j\in\N}{\prod} \mathbb{P}( \vert  X_{j,k}\vert \leq 2^{\alpha j} A )^{2^j},
\end{align}
where we have used the independence of the $c_{j,k}$ in the second line and that the $X_{j,k}$ are identically distributed in the last one. 
We are interested in the tail regime, specifically when $A$ takes very small or very large values for two main reasons:
\begin{itemize}
\item \textbf{Relevance to Applications}: As detailed in Section \ref{spect-estimation} examining small values of $A$ is particularly relevant for applications.  In such cases, wavelet leaders are  preferred over wavelet coefficients as multiresolution quantities in scaling functions. This preference arises because, unlike wavelet coefficients, the distributions of wavelet leaders vanish near zero, enabling the reliable numerical computation of negative moments which, in some applications yields a key information on the decreasing part of the multifractal spectrum. The computations presented here represent the first mathematical effort to rigorously substantiate this claim.
\item \textbf{Exponential Decay of Tails}: Log-concave distributions are defined by tails that decay at least as rapidly as exponential distributions. While the explicit asymptotic results we obtain cannot definitively confirm or refute the log-concavity of the distributions of log-leaders, they do show that the tails of these distributions decay exponentially fast. This finding is consistent with the hypothesis of log-concavity tested in the previous section.
\end{itemize}
\color{black}
The following approximation results yield  the tail distributions of log-leaders.



%
\begin{theorem}
\label{Th: small A leaders}
let  $X=(X(t))_{t\in\R_+}$ denote  the random wavelet series defined by \eqref{Eq: decomp wavelet process} where the distribution of the $X_{j,k}$ is given by \eqref{Eq: density generalized Gaussian}; then the following results hold:  
\begin{enumerate}
\item \textbf{\underline{Distribution of small leaders}} Assume that $\alpha > \log(1.13\pi) / \log 4\simeq 0.914$. Let $A\leq 2^{-\alpha }$. 
Then, there exist constants $\lambda_1,\lambda_2>0$ such that
\begin{equation}\label{Eq: tail bound small A}
\lambda_1\frac{2^{\alpha}}{A} \exp\Big(A^{-1/\alpha}\log\Big(\frac{2c_{l_\beta} \kappa_\beta\Lambda_\beta}{ 2^{2\alpha}}\Big)\Big)\leq \mathbb P( \ell^1_{0,0} \leq A)\leq \lambda_2 \frac{2^{\alpha}}{A}\exp\Big(A^{-1/\alpha}\log\Big(\frac{2c_{l_\beta} \kappa_\beta\Lambda_\beta}{ 2^{2\alpha}}\Big)\Big),
\end{equation}
where 
\begin{equation*}
c_{l_\beta}=\prod_{j\geq l_\beta}(1-1/(4j^2)),
\end{equation*}
$l_\beta>0$ is a constant, $\kappa_\beta$ and $\Lambda_\beta\in(1,\pi/2)$ are given by \eqref{Eq: density generalized Gaussian} and \eqref{Eq: def lambda beta}, respectively.
\item \textbf{\underline{Distribution of large leaders}}. Assume that $\alpha$ satisfies $2^{\alpha\beta}(1/\log(2)-\alpha\beta+\log_2(\alpha\beta\log(2)))-1>0$. For any $\beta>0$ there exist $A_\beta$ and $C_{\alpha,\beta}>0$ such that for all $A>A_\beta$,
\begin{equation}\label{Eq: tail bound large A}
\mathbb{P}(\ell^1_{0,0} > A)\leq \frac{C_{\alpha,\beta}\e^{-A^\beta}}{\beta A^{\beta-1}}.
\end{equation}
\vspace{3pt}\\
\textbf{Case 1:}  If  $\beta>1$, then  $A_\beta$ can be chosen such that
\begin{equation*}
\frac{1}{1+(\beta-1)/(\beta A_{\beta}^\beta)}=0.99.
\end{equation*}
\textbf{Case 2: } If  $0<\beta<1$, then  $A_\beta$ can be chosen such that
\begin{equation*}
\frac{1}{1-(1-\beta)/(\beta A_\beta^\beta)}=1.01.
\end{equation*}
\textbf{Case 3: $\beta=1$.}  If  $A$ is chosen large enough so that 
$ \forall j\in \N$, $2^{\alpha j}A\geq A+2^{\alpha}j$, 
then
\begin{equation*}
\mathbb{P}(\ell^1_{0,0} > A)\leq \e^{-j2^{\alpha}}=\e^{-A}\frac{1}{1-\frac{2}{e^{2^\alpha A}}}.
\end{equation*}
\end{enumerate}
\end{theorem}
%
\begin{remark}[Conditions on $\alpha$]\label{Rmk:alpha} For technical reasons and to facilitate calculations, we impose conditions on the value of $\alpha$ that depend on whether we are analyzing the distribution of small or large leaders.  
\\
The function $ \beta \in \mathbb{R}_{+}^* \mapsto \kappa_\beta = \beta / \Gamma(1/\beta) $ is bounded by $1.13$. Consequently, when considering small leaders, we choose \[  \alpha > \log(1.13\pi) / \log 4 \simeq 0.914 . \] Indeed, since $ c_{l_\beta} \in (2/\pi,1) $ and $ \Lambda_\beta \in (1,\pi/2) $, we obtain 
\begin{equation}\label{Eq: bound kappa beta}
2\alpha\log 2 > \log(1.13\pi) > \log\big(2 c_{l_\beta} \kappa_\beta \Lambda_\beta\big).
\end{equation}
For large leaders, we consider the function $ g_c: x \in \mathbb{R}^+ \mapsto 2^{cx} - 2^c x-1$ for any $ c > 0 $. It can be shown that if $ g(x_c) > 0 $, where  
\begin{equation*}
x_c = \frac{2^c}{\log 2} - c + \log_2(c \log 2),
\end{equation*}
then $g_c$ remains positive over $ \mathbb{R}_+ $. Let $c=\alpha\beta$. Assuming that $ \alpha $ satisfies $x_{\alpha\beta}>0$, we further have that for any $ j \geq 0 $, $\e^{-(2^{\alpha j} A)^\beta} \leq \e^{-A^\beta}\e^{-j(2^{\alpha} A)^\beta}$, which is useful for computations.
\end{remark}

\begin{remark}[Order of distribution tails]
It follows from  \eqref{Eq: tail bound small A} and \eqref{Eq: bound kappa beta} that 
%
\begin{equation*}
\mathbb P(\ell^1_{0,0} \leq A)=\Theta\big(A^{-1} \exp\big(-c_{\alpha,\beta}A^{-1}\big)\big),
\end{equation*}
where for two real functions $f$ and $g$, $f(A)=\Theta (g(A))$ means that $f$ grows asymptotically (with respect to $A$) at the same rate as $g$; and $c_{\alpha,\beta}=2\alpha\log 2-\log(2 c_{l_\beta}\lambda_\beta\kappa_\beta)$ is a constant which does not  depend on $A$. The tail of the distribution for small log-leaders is exponential, while the factor $A^{-1}$ corresponds to a logarithmic correction. Equation \eqref{Eq: tail bound large A} shows that the tails of large leaders also exhibit exponential decay.
Both tail behaviors of large and small leaders are consistent with the properties of log-concave distributions, which exhibit exponential or sub-exponential decay.
\end{remark}
\color{black}
\noindent
The proof of Theorem \ref{Th: small A leaders} relies on Lemma \ref{Lem: Mills ratio} (see Appendix \ref{App: Mills ratio}). We will also use the following well-known identities
\begin{equation}\label{Eq: sum power}
\sum_{k=0}^Nkx^k=\frac{x\big(Nx^{N+1}-(N+1)x^N+1\big)}{(x-1)^2},
\end{equation}
that holds for any $N\in\mathbb N^*$ and $x\in\R$, as well as

\begin{equation}\label{Eq: infinite products}
 \underset{j\in\N^*}{\prod}\Big(1-\frac{1}{4j^2}\Big)=\frac{2}{\pi} \quad\text{and}\quad \underset{j\in\N^*}{\prod}\Big(1+\frac{1}{4j(j+1)}\Big)=\frac{4}{\pi}.
\end{equation}

\begin{proof}[Proof of Theorem \ref{Th: small A leaders}] \textbf{\underline{Distribution of small leaders}} Let some $\alpha>\log(1.13\pi)/\log 4$. Assuming that $A\leq 2^{-\alpha}$, there exists $l\in\N^*$ such that $2^{-\alpha (l+1)}<A\leq 2^{-\alpha l}$. For $t$ small enough,
\begin{equation*} 
\mathbb{P}( \vert X_{j,k} \vert \leq t)=\kappa_\beta \int_{-t}^{t} \e^{-| x|^\beta} dx  \underset{t}{\sim} \kappa_\beta\int_{-t}^{t} dx = 2\kappa_\beta t,
\end{equation*}
where the symbol $\sim$ indicates that the functions are equivalent in the limit when $t\rightarrow 0$.
Then,
\begin{align}\label{Eq: product small A}
\nonumber
\underset{j\in\N}{\prod} \mathbb{P}( \vert  X_{j,k}\vert \leq 2^{\alpha j}  A )^{2^j}
&\leq \prod_{j=0}^{l-1} \mathbb{P}( \vert  X_{j,k}\vert \leq 2^{\alpha (j-l)})^{2^j} \prod_{j\geq l}\mathbb{P}( \vert  X_{j,k}\vert \leq 2^{\alpha (j-l)})^{2^j} \\
\nonumber
&\sim \big(2\kappa_\beta\big)^{2^l-1} \prod_{j=0}^{l-1}2^{-2^{j}\alpha (l-j)}\prod_{j\geq l}\mathbb{P}( \vert  X_{j,k}\vert \leq 2^{\alpha (j-l)})^{2^j}\\
\nonumber
&=\big(2\kappa_\beta\big)^{2^l-1}\cdot 2^{-\alpha \sum_{j=0}^{l-1}2^j(l-j)}\prod_{j\geq l}\mathbb{P}( \vert  X_{j,k}\vert \leq 2^{\alpha (j-l)})^{2^j}\\
\nonumber
&= \big(2\kappa_\beta\big)^{2^l-1}\cdot 2^{-\alpha(2^{l+1}-(l+2))}\prod_{j\geq l}\mathbb{P}( \vert  X_{j,k}\vert \leq 2^{\alpha (j-l)})^{2^j}\\
&=\big(2\kappa_\beta\big)^{2^l-1}\cdot 2^{-\alpha(2^{l+1}-(l+2))}\prod_{i\geq 0}\mathbb{P}( \vert  X_{i,k}\vert \leq 2^{\alpha i})^{2^{i+l}},
\end{align}
where \eqref{Eq: sum power} was used in the third line with $x=1/2$ to show that
\begin{equation*}
\sum_{j=0}^{l-1}2^j(l-j)=2^l\sum_{k=1}^{l}k2^{-k}=2^l\cdot 2(l2^{-(l+1)}-(l+1)2^{-l}+1)=2^{l+1}-(l+2).
\end{equation*}
Similarly, we can check that
\begin{equation}\label{Eq: product small A bis}
\underset{j\in\N}{\prod} \mathbb{P}( \vert  X_{j,k}\vert \leq 2^{\alpha j}  \ A )^{2^j}>\big(2\kappa_\beta\big)^{2^{l+1}-1}\cdot 2^{-\alpha(2^{l+2}-(l+3))}\prod_{i\geq 0}\mathbb{P}( \vert  X_{i,k}\vert \leq 2^{\alpha i})^{2^{i+l+1}}.
\end{equation}
\noindent
Moreover, using \eqref{Eq: Asymptotics generalised gaussian}, for $x$ large enough, we have $\mathbb P(|X_{j,k}|\leq x)\underset{x}\sim 1-2\kappa_\beta\e^{-|x|^\beta}/(\beta x^{\beta-1})$. Then, there exists $l_{\beta}>0$ such that for any $i\geq l_\beta$,
\begin{equation}\label{Eq:lbeta}
1-\frac{1}{4i^2}\leq \bigg(1-\frac{2\kappa_\beta \e^{-|2^{\alpha i}|^\beta} }{\beta \big(2^{\alpha i}\big)^{\beta-1}}\bigg)^{2^{i+l+1}},
\end{equation}
as well as
\begin{equation}\label{Eq: inequality probability j small A}
\begin{aligned}
& \Big(1-\frac{1}{4i^2}\Big)\bigg(1-\frac{2\kappa_\beta \e^{-|2^{\alpha i}|^\beta} }{\beta \big(2^{\alpha i}\big)^{\beta-1}}\bigg)^{2^{i+l+1}} \leq \mathbb{P}( \vert  X_{i,k}\vert \leq 2^{\alpha i})^{2^{i+l+1}} \\
& \text{and} \\
& \mathbb{P}( \vert  X_{i,k}\vert \leq 2^{\alpha i})^{2^{i+l}} \leq \Big(1+\frac{1}{4i(i+1)}\Big)\bigg(1-\frac{2\kappa_\beta\e^{-|2^{\alpha i}|^\beta}}{\beta \big(2^{\alpha i}\big)^{\beta-1}}\bigg)^{2^{i+l}}.
\end{aligned}
\end{equation}

Using \eqref{Eq: infinite products}, this entails
\begin{align*}
c_{l_\beta}\prod_{i\geq l_\beta}\bigg(1-\frac{2\kappa_\beta \e^{-|2^{\alpha i}|^\beta} }{\beta \big(2^{\alpha i}\big)^{\beta-1}}\bigg)^{2^{i+l+1}}\leq\prod_{i\geq l_\beta}\mathbb{P}( \vert  X_{i,k}\vert \leq 2^{\alpha i})^{2^{i+l+1}} \\ \quad\text{and}\quad\prod_{i\geq l_\beta}\mathbb{P}( \vert  X_{i,k}\vert \leq 2^{\alpha i})^{2^{i+l}}\leq C_{l_\beta}\prod_{i\geq l_\beta}\bigg(1-\frac{2\kappa_\beta \e^{-|2^{\alpha i}|^\beta} }{\beta \big(2^{\alpha i}\big)^{\beta-1}}\bigg)^{2^{i+l}}
\end{align*}
with $c_{l_\beta}=\prod_{i\geq l_\beta}(1-1/(4i^2))\in(\pi/2,1)$ and $C_{l_\beta}=\prod_{i
\geq l_\beta}(1+1/(4i(i+1)))\in(1,4/\pi)$. We finally obtain
\begin{equation}\label{Eq: ineq small leaders}
\begin{aligned}
c_{l_\beta} k_\beta\prod_{i=0}^{l_\beta-1}\mathbb{P}( \vert  X_{i,k}\vert \leq 2^{\alpha i})^{2^{i+l+1}}\prod_{i\geq l_\beta}\bigg(1-\frac{2\kappa_\beta \e^{-|2^{\alpha i}|^\beta} }{\beta \big(2^{\alpha i}\big)^{\beta-1}}\bigg)^{2^{i+l+1}}\leq \mathbb P(\ell^1_{0,0}\leq A)
\\ \quad\text{and}\quad
\mathbb P(\ell^1_{0,0}\leq A)\leq
C_{l_\beta}K_\beta\prod_{i=0}^{l_\beta-1}\mathbb{P}( \vert  X_{i,k}\vert \leq 2^{\alpha i})^{2^{i+l}}\prod_{i\geq l_\beta}\bigg(1-\frac{2\kappa_\beta \e^{-|2^{\alpha i}|^\beta} }{\beta \big(2^{\alpha i}\big)^{\beta-1}}\bigg)^{2^{i+l}}
\end{aligned}
\end{equation}
with $k_\beta=\big(2\kappa_\beta\big)^{2^{l+1}-1}\cdot 2^{-\alpha (2^{l+2}-(l+3))}$ and $K_\beta=\big(2\kappa_\beta\big)^{2^l-1}\cdot 2^{-\alpha (2^{l+1}-(l+2))}$ come from \eqref{Eq: product small A} and \eqref{Eq: product small A bis}. Note that for $\tau\in\{0,1\}$, by \eqref{Eq:lbeta}, there exists $\Lambda_\beta\in(1,\pi/2)$ such that
\begin{equation}\label{Eq: def lambda beta}
\prod_{i\geq l_\beta}\bigg(1-\frac{2\kappa_\beta \e^{-|2^{\alpha i}|^\beta} }{\beta \big(2^{\alpha i}\big)^{\beta-1}}\bigg)^{2^{i+l+\tau}}= \bigg[\Lambda_\beta\prod_{i\geq l_\beta}\bigg(1-\frac{1}{4i^2}\Big)\bigg]^{2^{l+\tau}}=\big(c_{l_\beta}\Lambda_\beta\big)^{2^{l+\tau}}>0,
\end{equation}
so that the infinite products appearing on the left and right sides of \eqref{Eq: ineq small leaders} are equal to  strictly positive constants.
Inserting the previous results into \eqref{Eq: product small A} and \eqref{Eq: product small A bis}, we find that 
%
%
\begin{equation*}
\lambda_12^{\alpha (l+3)}\Big(\frac{2 c_{l_\beta} \kappa_\beta\Lambda_\beta}{2^{2\alpha}}\Big)^{2^{l+1}} \leq  \mathbb P( \ell^1_{0,0} \leq A)\leq \min \bigg\{\lambda_22^{\alpha (l+2)}\Big(\frac{2c_{l_\beta}\kappa_\beta\Lambda_\beta}{2^{2\alpha}}\Big)^{2^l}\,;\, 1\bigg\},
\end{equation*}
where $\lambda_1$ and $\lambda_2$ are positive constants proportional to
\begin{equation*}
\lambda_1=(2\kappa_\beta)^{-1}k_1c_{l_\beta} \prod_{i=0}^{l_\beta-1}\mathbb{P}( \vert  X_{i,k}\vert \leq 2^{\alpha i})^{2^{i+l+1}}\quad\text{and}\quad \lambda_2=(2\kappa_\beta)^{-1} k_2C_{l_\beta}\prod_{i=0}^{l_\beta-1}\mathbb{P}( \vert  X_{i,k}\vert \leq 2^{\alpha i})^{2^{i+l}}
\end{equation*}
Noting by \eqref{Eq: bound kappa beta} that $ 2 c_{l_\beta}\kappa_\beta\Lambda_\beta/2^{2\alpha}<1$,
\begin{equation*}
\lambda_1\frac{2^{\alpha}}{A} \exp\Big(A^{-1/\alpha}\log\Big(\frac{2c_{l_\beta} \kappa_\beta\Lambda_\beta}{ 2^{2\alpha}}\Big)\Big)\leq \mathbb P(\ell^1_{0,0} \leq A)\leq \lambda_2 \frac{2^{\alpha}}{A}\exp\Big(A^{-1/\alpha}\log\Big(\frac{2c_{l_\beta} \kappa_\beta\Lambda_\beta}{ 2^{2\alpha}}\Big)\Big),
\end{equation*}
as $A\in(2^{-\alpha(l+1)},2^{-\alpha l}]$.
Hence the result.
\color{black} 
\noindent
 \\
\color{black}
\vspace{3pt}\\
\noindent
\textbf{\underline{Distribution of large leaders}} Let $\beta\in\R_+^*$. We are interested in the behaviour of large log-leaders, i.e. we aim at approximating $\mathbb P(\ell^1_{0,0} >A)$ for $A>0$ sufficiently large.
First,
\begin{align} \label{Eq: distribution leaders product large}
\nonumber
\mathbb{P}( \ell^1_{0,0}> A) 
&= \mathbb{P}\Big( \underset{ \lambda \subset  I }{\sup} \vert c_{\lambda} \vert >A \Big) \\
\nonumber
&\leq \sum_{\lambda \subset I} \mathbb{P} \big( \vert c_{\lambda} \vert >A \big)\\
\nonumber
&= \underset{ (j,k):\lambda \subset I}{\sum} \mathbb{P}(\vert c_{j,k} \vert > A )\\
\nonumber
&= \underset{(j,k):\lambda \subset I}{\sum} \mathbb{P}( \vert  X_{j,k} \vert > 2^{\alpha j} A )\\
&=\underset{j\in\N}{\sum}2^j \mathbb{P}( \vert  X_{j,k} \vert > 2^{\alpha j}  \ A ),
\end{align}
where we have used that the $X_{j,k}$ are identically distributed in the last line. Let $A$ be sufficiently large, meaning any $A > A_\beta$, where $A_\beta > 0$ is chosen such that $0.99\kappa_\beta \frac{\e^{-A_\beta^\beta}}{\beta A_\beta^{\beta-1}} \leq \mathbb{P}(X_{0,k} > A_\beta) \leq 1.01\kappa_\beta \frac{\e^{-A_\beta^\beta}}{\beta A_\beta^{\beta-1}}$
(see \eqref{Eq: Mills ratio}). Then, for all $j\in\N$, 
\begin{equation}\label{Eq: inequality probability j}
0.99\cdot 2^j\bigg(\frac{\kappa_\beta \e^{-(2^{\alpha j}A)^\beta} }{\beta \big(2^{\alpha j} A\big)^{\beta-1}}\bigg)\leq2^j\mathbb{P}( \vert  X_{j,k}\vert >2^{\alpha j} A )\leq 1.01\cdot 2^j\bigg(\frac{\kappa_\beta\e^{-(2^{\alpha j}A)^\beta}}{\beta \big(2^{\alpha j} A\big)^{\beta-1}}\bigg).
\end{equation}
Note that (see Remark \ref{Rmk:alpha}), since $2^{\alpha\beta}(1/\log(2)-\alpha\beta+\log_2(\alpha\beta\log(2)))-1>0$, it follows that for any $j\geq 0$, $ \e^{-(2^{\alpha j}A)^\beta}\leq \e^{-A^\beta}\e^{-j(2^{\alpha}A)^\beta}$.
Combining this with \eqref{Eq: distribution leaders product large} and \eqref{Eq: inequality probability j} entails
\begin{align}\label{Eq: bound proba A large}
\mathbb{P}(\ell^1_{0,0} > A)\leq \frac{1.01\kappa_\beta \e^{-A^\beta}}{\beta A^{\beta-1}}
\underset{j\geq 0}{\sum}\frac{\e^{-j(2^{\alpha}A)^\beta}\cdot 2^j}{2^{\alpha(\beta-1)j}}\leq\frac{1.01\kappa_\beta \e^{-A^\beta}}{\beta A^{\beta-1}}\frac{1}{1-\frac{2^{1-\alpha(\beta-1)}}{\e^{2^{\alpha\beta}A^\beta}}}.
\end{align}
It is then enough to determine $A_\beta$ according to the values of $\beta\in(0,1)$.\\
\textbf{Case 1: $\beta>1$.} Using \eqref{Eq Mills beta grand}, we can choose $A_{\beta}$ such that
\begin{equation*}
0.99=\frac{1}{1+(\beta-1)/(\beta A_{\beta}^\beta)},
\end{equation*}
which enables us to satisfy \eqref{Eq: inequality probability j} for $j = 0$. We can check that \eqref{Eq: inequality probability j} holds for any $A>A_\beta$ and for all $j\in\N$. 
\vspace{3pt}\\
\noindent
\textbf{Case 2: $0<\beta<1$.} 
Using \eqref{Eq Mills beta petit}, we can choose $A_{\beta}$ such that
\begin{equation*}
\frac{1}{1-(1-\beta)/(\beta A_\beta^\beta)}=1.01,
\end{equation*}
which is \eqref{Eq: inequality probability j} for $j=0$. As before, we can check that \eqref{Eq: inequality probability j} holds for any $A>A_\beta$ and for all $j\in\N$. 
\vspace{3pt}\\
\noindent
\textbf{Case 3: $\beta=1$.} In that case, using that $\mathbb{P}( \vert  X_{j,k} \vert > 2^{\alpha j}A )=\e^{-2^{\alpha j}A}$, we get by \eqref{Eq: distribution leaders product large} that
\begin{align*}
\mathbb{P}(\ell^1_{0,0} > A)
&\leq \sum_{j\in\N}2^j \e^{-2^{\alpha j}A}\\
&\leq \e^{-A}\sum_{j\in\N}2^j \e^{-j2^{\alpha}A}=\e^{-A}\frac{1}{1-\frac{2}{e^{2^\alpha A}}},
\end{align*}
using that we have chosen $\alpha$ (see Remark \ref{Rmk:alpha}) such that $2^{\alpha j}A\geq A(1+2^{\alpha}j)$ for all $j\in\N$.
\end{proof}

\begin{example} Let us choose $\beta=2$. Then, the $X_{j,k}$ are Gaussian random variables with variance $1/2$ and $\kappa_\beta=1/\sqrt{\pi}$.\\
\textbf{\underline{Distribution of small leaders}} 
To obtain an expression of \eqref{Eq: tail bound small A} with explicit values, we need to choose $\alpha$. For example, let $\alpha=1$, and  $l_2=3$, so that $0.99\e^{-x^2}/(2\sqrt{\pi}x) \leq \mathbb P(X>x)\leq \e^{-x^2}/(2\sqrt{\pi}x)$ for all $x\geq 2^{l_2}$. 
Using \eqref{Eq: product small A} as well as $c_{l_2}=\prod_{j\geq 3}(1-1/(4j^2))=128/(45\pi)$, we obtain that  and
\begin{equation}\label{Eq: A small Gaussian}
\mathbb{P}(\ell^1_{0,0} \leq  A)=\Theta \bigg(2A^{-1} \exp\Big(A^{-1}\Big(\log\Big(\frac{2c_{l_2}\Lambda_2}{\sqrt{\pi}}\Big)\Big)-2 \log(2)\Big)\bigg).
\end{equation}
where $c_{l_2}\Lambda_2\in(256/(45\pi^2),1)$. We can check that $2c_{l_2}\Lambda_2/\sqrt{\pi}-2\log(2)<2/\sqrt{\pi}-2\log 2<0$.\\
%
%
\noindent
\textbf{\underline{Distribution of large leaders}} 
 we can choose $A_{2}$ such that
\begin{equation*}
0.99=\frac{1}{1+1/(2 A_{2}^2)}\;\text{i.e.}\; A_2=\sqrt{99/2}\simeq 7.04.
\end{equation*}
Then for all $A>A_2$, by \eqref{Eq: tail bound large A}, there exists $c>0$ such that
\begin{equation}\label{Eq: A_2 Gaussian}
\mathbb{P}( \ell^1_{0,0} > A)\leq \frac{c\e^{-A^2}}{A}\frac{1}{1-\frac{2^{1-\alpha}}{\e^{2^{2\alpha}A^2}}}.
\end{equation}
\end{example}

\begin{remark}
Both results \eqref{Eq: A small Gaussian} and \eqref{Eq: A_2 Gaussian} align with the negative outcomes of the log-normality test performed on the leaders (see Section \ref{Subsec: Normality}): even when the distribution of the variables $X_{j,k}$ is Gaussian, the tails of the distribution of the log-leaders are significantly lighter than those of a Gaussian.
\end{remark}

\color{blue}
\color{black}
\section{Conclusion and prospects}\label{Sec: conclusion}
In this work, we focused on the distribution of log-leaders, key quantities for estimating the multifractal spectrum. The first step involved rejecting the commonly held assumption that log-leaders follow Gaussian distributions. A normality test was applied to the log-leaders of various types of processes (fractional Brownian motion, compound Poisson process, multifractal random walk) as well as to heart rate and speed data from marathon runners. We then applied an experimentally validated test for log-concavity, which allowed us to propose a wider nonparametric model for  log-leaders:  log-concave distributions. This allows for a reassessment of the estimation of the parameters $c_1$ and $c_2$ and for establishing a confidence interval for the estimation of $c_1$ based on the CLT, without relying on the Gaussianity assumption.
Finally, we conducted  a preliminary theoretical study which allowed us to  approximate the cumulative distribution function of the log-leaders for random wavelet series, assuming independence of wavelet coefficients both within and across scales. This corroborates the empirical discovery that, although the distributions of log-leaders are not normal, they belong to the larger class of  log-concave distributions. This exploratory work should help refine the methods used in multifractal analysis until now, which seem numerically justified but whose validity is demonstrated under the overly restrictive/invalidated assumption of normality.\\
\noindent
A primary direction for applying and extending this work is in estimating the scaling exponents that characterize scale invariance properties, as well as constructing confidence intervals to measure their quality. Standard methods for constructing such bounds are often based on Gaussian theory; they prove effective for Gaussian self-similar processes \cite{abry2000wavelets,veitch1999wavelet} but perform poorly when applied to multifractal processes due to their non-Gaussian nature. However, the use of nonparametric or Empirical Cumulative Distribution Function bootstrap (ECDF-bootstrap) allows for circumventing the problem by assuming that little to nothing is known about the underlying model of the data (see \cite{wendt2007bootstrap}). To take advantage of the fact that the distributions of the log-leaders are log-concave, it would be interesting to explore other types of methods, such as those studied by Azadbakhsh, Jankowski and Gao \cite{azadbakhsh2014computing}, for computing confidence intervals for log-concave densities. These methods are notably based on the fact that the log-concave MLE can be used in Monte Carlo bootstrap procedures, as noted by Cule and Samworth \cite{Cule_Samworth_2010}. We have provided a confidence interval for $c_1$ based on the CLT. Establishing one for $c_2$ would require computing a confidence interval for the variance of log-concave distributions—a challenging problem that remains unsolved.\\
Furthermore, our result nuances that of  \cite{wendt2013bayesian}, which state that the distribution of the log-leaders can be approximated by a Gaussian. Based on this assumption, Combrexelle,  Wendt, Dobigeon, Tourneret, McLaughlin, and Abry construct a semi-parametric model for the statistics of the logarithm, which they then use to define and study a Bayesian estimator of the multifractality parameter for synthetic multifractal processes in \cite{wendt2013bayesian}, image texture \cite{combrexelle2015bayesian} and multivariate time series in \cite{Combrexelle2016EUSIPCO}.
From this perspective, it would be interesting to extend or adapt the semi-parametric model and Bayesian procedure  they introduced by considering that the distributions of the log-leaders are log-concave rather than assuming that they can be approximated by a Gaussian random variable.   \\ Note also that, in the present article, we only considered  the two  major multifractality parameters used for classification which are deduced form wavelet leaders, i.e. $c_1$  and $c_2$. The third  major multifractality parameter, namely,  the uniform regularity exponent $H^{\min}_f$, is derived directly from wavelet coefficients; its statistical estimation will be addressed in \cite{BNHHSJ2}.  \\
Finally, it would be valuable to relax the assumption of independence between scale coefficients, as used in the theoretical study of Section \ref{Sec: theoretical study}, and to develop a model that addresses the dependencies highlighted numerically (see for instance \cite{Morel_etal_2022_dependencies}).  A first question is to determine which results concerning the multifractal analysis of random wavelet series  may be modified if one drops the assumption of independence between scale; this is one of the topics addressed in the forthcoming paper \cite{BNHHSJ2}, where we will focus on the important example  where the laws of wavelet coefficients are  Gaussian mixtures. \\
Another important issue is  to extend this analysis to $p$-leaders, for which the supremum of wavelet coefficients  in the definition of wavelet leaders is replaced by a  $l^p$-norm.  Indeed the use of $p$-leaders has a wider range of validity (leaders require a uniform regularity assumption which is often not met in practice, see \cite{Jaffard2015}) and furthermore, it has been noticed that  a multifractal  analysis based on $p$-leaders for $p$ close to 2 enjoys better statistical properties than when using wavelet leaders \cite{PART1,PART2}.
\backmatter
\bmhead{Acknowledgements} The authors are very grateful to  V\'eronique Billat from the  IBISC laboratory at  Universit\'e d'\'Evry Paris-Saclay for allowing them to use the marathon runners data which she and her team collected.  

\begin{appendices}
\section{Equality between the 1-leader and the 3-leader scaling functions}\label{App: equality leaders}

\begin{proof}[Proof of Proposition \ref{equalscal}]
For simplicity, we prove the result in the one-variable case, noting that it extends naturally to the general case.
\noindent
First, since for all $\lambda\geq 0$, $\ell^1_{\lambda }\leq \ell^3_{\lambda }$, it follows that all $q>0$ and $j$, $S_1(j,q)\leq S_3(j,q)$,
so that 
\[ \forall q >0, \qquad \zeta_1(q)\geq \zeta_3(q). \]
Similarly, we have that for all $q<0$ and $j$, $S_1(j,q)\geq S_3(j,q)$ so that 
\[ \forall q <0, \qquad \zeta_1(q)\leq \zeta_3(q). \]
In order to obtain the reversed inequalities, we note that 
$ \ell^3_{\lambda } = \max  ( \ell^1_{\lambda^- }, \ell^1_{\lambda }, \ell^1_{\lambda^+ } ) $,
where $\lambda^-$ and $\lambda^+$ are the two dyadic intervals neighbours to  $\lambda$ and of the same generation. It follows that for all $q>0$,  $S_3(j,q)\leq  3 S_1(j,q)$, so that 
\[ \forall q >0, \qquad \zeta_3(q)\geq \zeta_1(q). \]
Assume now that $q <0$. We denote by $\lambda'(\lambda) $ and $\lambda''(\lambda) $ the two sub-intervals of $\lambda$  of generation $j+2$ which are included in the interior of $\lambda$ ($\lambda'(\lambda) $ being on the left of $\lambda''(\lambda) $). it follows that $3 \lambda'(\lambda) \subset \lambda$, so that  
\[ \ell^3_{\lambda'(\lambda) } = \underset{ \lambda'' \subset  3\lambda'(\lambda)}{\sup} \vert c_{\lambda'' }\vert \leq \underset{ \lambda'' \subset  \lambda}{\sup} \vert c_{\lambda'' }\vert = \ell^1_{\lambda }. \] 
Denote by $\Lambda_j$ the set of dyadic intervals of generation $j$ (i.e. of length $2^{-j}$). It follows that
\[  \forall q <0, \qquad 
\sum_{\lambda\in {\Lambda_j}} \vert \ell^1_{\lambda} \vert^q \leq \sum_{\lambda\in {\Lambda_j}} \vert \ell^3_{\lambda'(\lambda)} \vert^q \leq \sum_{\lambda'\in {\Lambda_{j+2}}} \vert \ell^3_{\lambda'} \vert^q .  \]
It follows that 
\[  \forall q <0, \qquad  S_1(j,q)\leq  4 \cdot  S_3(j+2,q), \]
so that 
\[ \forall q <0, \qquad \zeta_1(q)\geq \zeta_3(q). \]
Putting together the four inequalities we have proved between values of the scaling functions yields \eqref{equalzeta}. 
\end{proof}

\section{General bounds on Mills ratio}\label{App: Mills ratio}


In probability theory, the Mills ratio \cite{Mills_1926} states that for a continuous real random variable $X$ with density $f$ and for any $x\in\R$,
\begin{equation}\label{Eq: Mills ratio}
\frac{I(x)}{f(x)}=\lim_{\varepsilon\to 0}\frac{1}{\varepsilon}\mathbb P(x<X\leq x+\varepsilon|X>x),
\end{equation}
where for all $x\in\mathbb R$, we define $I(x)=\mathbb P(X>x)$.
Bounding \eqref{Eq: Mills ratio} provides insights on the distribution of the tails of a random variable. For instance, if $X$ has a standard normal distribution, for all $x>0$,
\begin{equation*}
\frac{xe^{-x^2/2}}{\sqrt{2\pi}(x^2+1)}\leq I(x)\leq \frac{e^{-x^2/2}}{\sqrt{2\pi}x}.
\end{equation*}
In the same spirit, in the following lemma, we provide  bounds for random variables having an exponential density. 
\begin{lemm}[General bounds for Mills ratio]\label{Lem: Mills ratio}
Let $g:\mathbb R\rightarrow\mathbb R$  be a twice differentiable positive function such that $g'$ is positive on $\R_+^*$ and $g''$ has a constant sign on $\mathbb{R}_+$. Assume, moreover, that there exists a function $M:\R_+\to\R_+$ such that for $x>0$, 
\begin{equation*}
\sup_{t\in[x,\infty)} \bigg|\frac{g''(t)}{(g'(t))^2}\bigg|\leq M(x).
\end{equation*}
Let $X$ be a real random variable with density $f=\kappa e^{-g}$, where $\kappa>0$ is a normalisation constant.
For all $x\in\mathbb R$, we define $I(x)=\mathbb P(X>x)=\int_x^{\infty}\kappa e^{-g(t)}dt$.
Then, 
\begin{equation}\label{Eq: Mills IPP}\forall x\in\R_+^*, \qquad 
\frac{f(x)}{g'(x)(1+M(x)\mathbf 1_{\{g''>0\}})}\leq I(x)\leq \frac{f(x)}{g'(x)(1-M(x)\mathbf 1_{\{g''<0\}})}.
\end{equation}
\end{lemm}
\begin{proof}
For $x\in\R_+^*$,
\begin{align*}
\kappa^{-1}I(x)
&=\int_x^\infty \frac{g'(t)}{g'(t)}e^{-g(t)}dt=\lim_{A\rightarrow\infty}\Big[-\frac{1}{g'(t)}e^{-g(t)}\Big]_x^{A}-\int_x^{\infty}\frac{g''(t)}{(g'(t))^2}e^{-g(t)}dt\\
&=\frac{1}{g'(x)}e^{-g(x)}-\int_x^{\infty}\frac{g''(t)}{(g'(t))^2}e^{-g(t)}dt.
\end{align*}
If $g''>0$, we get that for any $x\in\R_+^*$,
\begin{equation*}
 \frac{e^{-g(x)}}{g'(x)}-\kappa^{-1}M(x)I(x)\leq \kappa^{-1}I(x)\leq \frac{e^{-g(x)}}{g'(x)},
\end{equation*}
whereas if $g''<0$, 
\begin{equation*}
 \frac{e^{-g(x)}}{g'(x)}\leq \kappa^{-1}I(x)\leq \frac{e^{-g(x)}}{g'(x)}+\kappa^{-1}m(x)I(x).
\end{equation*}
The two inequalities lead to \eqref{Eq: Mills IPP}.
\end{proof}
\begin{example}
\begin{enumerate}
\item \textbf{Standard Gaussian.} Taking $g(x)=x^2/2$ and $\kappa=(2\pi)^{-1/2}$, we find $M(x)=1/x^2$. This boils down to the well-known result:
\begin{equation*}. \forall  x>0,  \qquad 
\frac{1}{\sqrt{2\pi}}\frac{\e^{-x^2/2}}{x(1+1/x^2)}\leq I(x)\leq \frac{1}{\sqrt{2\pi}}\frac{\e^{-x^2/2}}{x}.
\end{equation*}
\item \textbf{Generalized Gaussian with light tails ($\beta>1$).} Taking $g(x)=|x|^{\beta}$ and $\kappa_\beta=\beta/(2\Gamma(1/\beta))$ with $\beta>1$, clearly  $g'>0$ and $g''>0$ on $\R_+^*$. We have that for all $x> 0$, $M(x)=(\beta-1)/(\beta x^\beta)$ and 
\begin{equation}\label{Eq Mills beta grand}
\frac{1}{2\Gamma(1/\beta)}\frac{\e^{-x^\beta}}{ x^{\beta-1}(1+(\beta-1)/(\beta x^\beta))}\leq I(x)\leq \frac{1}{2\Gamma(1/\beta)}\frac{\e^{-x^\beta}}{ x^{\beta-1}}.
\end{equation}
\item \textbf{Generalized Gaussian with heavy tails ($0<\beta<1$).} Taking $g(x)=|x|^{\beta}$ and $\kappa=\beta/(2\Gamma(1/\beta))$ with $\beta<1$, we check that $g'>0$ and $g''<0$ on $\R_+^*$. We have that for all $x>0$, $M(x)=(1-\beta)/(\beta x^\beta)$ and that
\begin{equation}\label{Eq Mills beta petit}
\frac{x^{1-\beta}\e^{-x^\beta}}{2\Gamma(1/\beta)}\leq I(x)\leq \frac{1}{2\Gamma(1/\beta)}\frac{x^{1-\beta}\e^{-x^\beta}}{(1-(1-\beta)/(\beta x^\beta))}.
\end{equation}
\item \textbf{Laplace distribution ($\beta=1$).} Taking $g(x)=|x|$ and $\kappa=1/(2\Gamma(1))=1/2$ and 
\begin{equation}\label{Eq Mills beta 1}
\forall x>0,\qquad I(x)=\frac{1}{2}\e^{-x}=:f_1(x).
\end{equation}
\item \textbf{Asymptotics}. Given the density \eqref{Eq: density generalized Gaussian} of a generalized Gaussian random variable $X$, we have, asymptotically,
\begin{equation}\label{Eq: Asymptotics generalised gaussian}
\mathbb P(X>x)\underset{x\to\infty}\sim \frac{f_\beta(x)}{\beta x^{\beta-1}}.
\end{equation}
Note that, since $X$ is symmetric, we obtain, asymptotically,
\begin{equation*}
\mathbb P(|X|\leqslant x)=1-2\mathbb P(X>x)\underset{x}\sim 1-\frac{2f_\beta(x)}{\beta x^{\beta-1}}.
\end{equation*}
\end{enumerate}
\end{example}
\end{appendices}
\bibliography{sn-bibliography}

\end{document}